\documentclass[envcountsame,envcountsect]{llncs}
\usepackage[letterpaper,hmargin=1.5in,vmargin=1.5in]{geometry}

\usepackage{graphicx}
\usepackage{amssymb}
\usepackage{epstopdf}
\usepackage{amsmath}
\usepackage{qip}
\usepackage{framed}
\usepackage{epsfig}
\usepackage{color}
\usepackage[colorlinks=false]{hyperref}

\newcommand{\ot}{\ensuremath{\mbox{{\sc {\small 1-2}-ot}}}}

\newcommand{\otalone}{{\sc ot}}
\newcommand{\nlb}{{\sc nl}}

\newcommand{\emb}[1]{\ensuremath{\mathcal{E}}(#1)}

\newcommand{\expe}[1]{\ensuremath{\mbox{e}^{#1}}}
\newcommand{\dep}[2]{\ensuremath{#1 \searrow #2}}
\newcommand{\otp}[1]{\ensuremath{\ot^{#1}}}
\newcommand{\srotp}[1]{\ensuremath{\mbox{\sc rot}^{#1}}}
\newcommand{\srot}{\srotp{r}}
\newcommand{\sotr}{\otp{r}}

\newcommand{\sand}{\ensuremath{\mbox{\sc sand}}}

\newcommand{\otn}{\ensuremath{\ot_p}}

\newcommand{\pnl}{\ensuremath{P^{\mbox{{\tiny \nlb}}}_{X,Y}}}
\newcommand{\pot}{\ensuremath{P^{\mbox{{\tiny \otalone}}}_{X,Y}}}

\newcommand{\psot}{\ensuremath{P^{\mbox{{\tiny {\sc ot}$^r$}}}_{X,Y}}}
\newcommand{\potn}{\ensuremath{P^{\mbox{{\tiny \otalone$_p$}}}_{X,Y}}}
\newcommand{\potnq}{\ensuremath{P^{\mbox{{\tiny \otalone$_{1/4}$}}}_{X,Y}}}

\newcommand{\psrot}{\ensuremath{P^{\mbox{{\tiny {\sc rot}$^r$}}}_{X,Y}}}

\newcommand{\hil}{\mathcal{H}}

\newcommand{\unit}[1]{{\sf U}(#1)}
\newcommand{\idf}[1]{\ensuremath{{\mathsf{ID}}(#1)}}

\newcommand{\correct}{strictly correct}
\newcommand{\correctly}{strict-correctly}
\newcommand{\correctness}{strict correctness}

\newcommand{\CorrectNess}{Strict Correctness}

\newcommand{\assign}{\ensuremath{\kern.5ex\raisebox{.1ex}{\mbox{\rm:}}\kern -.3em =}}
\newcommand{\eps}{\varepsilon}

\newcounter{itm}

\title{
Quantifying the Leakage of Quantum Protocols\\for Classical Two-Party Cryptography\thanks{A previous version of this article as appeared at ASIACRYPT 2009~\cite{SSS09}.}
}
\author{
Louis Salvail\inst{1}
\and
Christian Schaffner\inst{2,3}
\and Miroslava Sot\'{a}kov\'{a}\inst{4}
}
\institute{
Universit\'e de Montr\'eal (DIRO), QC, Canada\\
\email{salvail@iro.umontreal.ca}
\and
Institute for Logic, Language and Computation (ILLC)\\ University of  Amsterdam, The Netherlands\\
\email{c.schaffner@uva.nl}
\and
Centrum Wiskunde \& Informatica (CWI), Amsterdam, The Netherlands\\
\and
Knewton, Inc, NY, USA\\
\email{gwhitehawk@gmail.com}
}
\begin{document}
\pagestyle{plain}
\maketitle
\begin{abstract}
  We study quantum protocols among two distrustful parties. By adopting a rather strict definition of correctness---guaranteeing that honest players
  obtain their correct outcomes only---we can show that every \correct\  quantum protocol implementing a non-trivial \emph{classical} primitive necessarily leaks information to a dishonest player. This extends known impossibility
  results to all non-trivial primitives. We provide a framework for
  \emph{quantifying} this leakage and argue that leakage is a good measure
  for the privacy provided to the players by a given protocol. Our
  framework also covers the case where the two players are helped by a
  trusted third party. We show that despite the help of a trusted
  third party, the players cannot amplify the cryptographic power of
  any primitive. All our results hold even against quantum
  honest-but-curious adversaries who honestly follow the protocol but
  purify their actions and apply a different measurement at the end of
  the protocol. As concrete examples, we establish lower bounds on the
  leakage of standard universal two-party primitives such as oblivious
  transfer.

\vspace{1mm}
{\bf Keywords:} two-party cryptography, quantum protocols, quantum
information theory, information leakage.

\end{abstract}

\section{Introduction}
\label{chap:intro}
Quantum communication allows to implement tasks which are classically
impossible. The most prominent example is quantum key
distribution~\cite{BB84} where two honest players establish a secure
key against an eavesdropper. In the two-party setting however, quantum
and classical cryptography often show similar limits. Oblivious
transfer~\cite{Lo97}, bit commitment~\cite{Mayers97,LC97}, and even fair
coin tossing~\cite{Kitaev03} are impossible to realize securely both
classically and quantumly.  On the other hand, quantum cryptography
allows for some weaker primitives impossible in the classical
world. For example, quantum coin-flipping protocols with maximum bias
of $\frac{1}{\sqrt{2}}-\frac12$ exist\footnote{In fact, protocols with
  better bias are known for weak quantum coin
  flipping~\cite{Mochon04,Mochon05,Mochon07}.} against any
adversary~\cite{CK09} while remaining impossible based solely on
classical communication. A few other weak primitives are known to be
possible with quantum communication.
For example, the generation of an additive
secret-sharing for the product $xy$ of two bits, where Alice holds bit
$x$ and Bob bit $y$, has been introduced by Popescu and Rohrlich as
machines modeling non-signaling non-locality (also called NL-boxes)~\cite{PR94}.
If Alice and Bob share an EPR pair, they can simulate an NL-box with symmetric
error probability $\sin^{2}{\frac{\pi}{8}}$~\cite{PR94,BLMPPR05}.
Equivalently, Alice and Bob can implement {\em 1-out-of-2 oblivious
  transfer} (\ot) privately provided the receiver Bob gets the bit of
his choice only with probability of error $\sin^{2}{\frac{\pi}{8}}$~\cite{Ambainis05OT}.  It is easy to verify that even with such imperfection these two
primitives are impossible to realize in the classical world.
This discussion naturally leads to the following question:
\begin{itemize}
\item Which two-party cryptographic primitives are possible to achieve using
quantum communication?  
\end{itemize}
Most standard classical two-party primitives have been shown
impossible to implement securely against weak quantum adversaries
reminiscent to the classical honest-but-curious (HBC)
behavior~\cite{Lo97}. The idea behind these impossibility proofs is to
consider parties that {\em purify} their actions throughout the
protocol execution.  This behavior is indistinguishable from the one
specified by the protocol but guarantees that the joint quantum state
held by Alice and Bob at any point during the protocol remains
pure. The possibility for players to behave that way in any two-party
protocol has important consequences. For instance, the impossibility
of quantum bit commitment follows from this fact~\cite{Mayers97,LC97}:
After the commit phase, Alice and Bob share the pure state
$\ket{\psi^x} \in \hil_A \otimes \hil_B$ corresponding to the
commitment of bit~$x$. Since a proper commitment scheme provides no
information about $x$ to the receiver Bob, it follows that
$\tr_A\proj{\psi^{{0}}} =\tr_A\proj{\psi^1}$. In this case, the
Schmidt decomposition guarantees that there exists a unitary $U_{0,1}$
acting only on Alice's side such that $\ket{\psi^1} = (U_{0,1}\otimes
\I_{B})\ket{\psi^{{0}}}$.  In other words, if the commitment is
concealing then Alice can open the bit of her choice by applying a
suitable unitary transform only to her part. A similar argument allows
to conclude that \ot\ is impossible~\cite{Lo97}: Suppose Alice is
sending the pair of bits $(b_0,b_1)$ to Bob through \ot.  Since Alice
does not learn Bob's selection bit, it follows that Bob can get bit
$b_0$ before undoing the reception of $b_0$ and transforming it into
the reception of $b_1$ using a local unitary transform similar to
$U_{0,1}$ for bit commitment. For both these primitives, privacy for
one player implies that local actions by the other player can
transform the honest execution with one input into the honest
execution with another input.

In this paper, we investigate the cryptographic power of two-party
quantum protocols against players that purify their actions while
trying to implement a classical primitive. This {\em
  quantum honest-but-curious (QHBC)} behavior is the natural quantum
version of classical HBC behavior.  
This class of adversaries  
was recently called \emph{(perfectly) specious} in~\cite{DNS10}.
It contains all adversaries that could \emph{prove} to a judge,
at any step during a protocol execution, that
the joint state (up to an adversary's local computation) is the honest one. We consider classical primitives providing Alice and Bob with random variable $X$
and $Y$ respectively according distribution $P_{X,Y}$.
Any such
$P_{X,Y}$ models a two-party cryptographic primitive where neither
Alice nor Bob provide input. For the purpose of this paper, this model
is general enough since any two-party primitive with inputs can be
randomized (Alice and Bob pick their input at random) so that its
behavior can be described by a suitable joint probability distribution
$P_{X,Y}$. If the classical primitive with inputs $f:A\times B \rightarrow W\times Z$ 
is implemented  securely by some protocol $\pi_f$  then it must also remain
secure when Alice's and Bob's private input $(a,b)\in_R {A}\times {B}$ 
is picked uniformly at random. In this case, the joint probability distribution 
$P_{X,Y}$ implemented by $\pi_f$ is simply:
\[    P_{X,Y}((a,w),(b,z)) = \frac{ \Pr{\left(f(a,b)=(w,z)\right)}}{|A| \cdot |B|}.
\]
 If the randomized version $P_{X,Y}$ is shown to be
impossible to implement securely by any quantum protocol then
the original primitive with inputs must also be impossible. 

Any quantum protocol implementing $P_{X,Y}$ must produce, when both
parties purify their actions, a joint pure state $\ket{\psi}\in
\hil_{AA'} \otimes\hil_{BB'}$ that, when subsystems of $A$ and $B$ are
measured in the computational basis, leads to outcomes $X$ and $Y$
according the distribution $P_{X,Y}$. Notice that the registers $A'$
and $B'$ only provide the players with extra working space and, as
such, do not contribute to the output of the functionality (so parties
are free to measure them the way they want). In this paper, we adopt a
somewhat strict point of view and define a quantum protocol $\pi$ for
$P_{X,Y}$ to be \emph{\correct}\ if and only if the correct outcomes
$X,Y$ are obtained \emph{and} the registers $A'$ and $B'$ do not
provide any additional information about $Y$ and $X$ respectively
since otherwise $\pi$ would be implementing a different primitive
$P_{XX',YY'}$ rather than $P_{X,Y}$. 
The state $\ket{\psi}$ produced by any \correct\ protocol for $P_{X,Y}$
is what we call a {\em quantum embedding} of $P_{X,Y}$.  An embedding is
called \emph{regular} if registers $A'$ and $B'$ are empty. Any
embedding $\ket{\psi}\in \hil_{AA'}\otimes\hil_{BB'}$ can be produced
in the QHBC model by the trivial protocol asking Alice to generate
$\ket{\psi}$ before sending the quantum state in $\hil_{BB'}$ to
Bob. It follows that in the QHBC model, any embedding of $P_{X,Y}$
corresponds to a \correct\  protocol and, since any protocol implementing  $P_{X,Y}$ can be purified 
in the bare model, any \correct\ protocol generates some embedding of $P_{X,Y}$
in the QHBC model.

Notice that if $X$ and $Y$ were provided privately to Alice and
Bob---through a trusted third party for instance---then the expected
amount of information one party gets about the other party's output is
minimal and can be quantified by the Shannon mutual information
$I(X;Y)$ between $X$ and $Y$. Assume that $\ket{\psi}\in
\hil_{AA'}\otimes\hil_{BB'}$ is an embedding of $P_{X,Y}$ produced by
a \correct\ quantum protocol.  We define the leakage of $\ket{\psi}$ as
\begin{equation}
\label{leak_intro} 
\Delta_{\psi} \assign \max \left\{ \, S(X;BB')-I(X;Y) \, ,
  \:S(Y;AA')-I(Y;X) \,\right\},
\end{equation}
where $S(X;BB')$ (resp. $S(Y;AA')$)  is the information
the quantum registers $BB'$ (resp. $AA'$) provide about 
the output $X$ (resp. $Y$). That is, the leakage is the
maximum amount of extra information about the other party's output
given the quantum state held by one party. It turns out that
$S(X;BB')=S(Y;AA')$ holds for all embeddings, exhibiting a symmetry
similar to its classical counterpart $I(X;Y)=I(Y;X)$ and therefore,
the two quantities we are taking the maximum of in~\eqref{leak_intro} coincide. 

\subsection{Contributions} Our first contribution establishes that
the notion of leakage is well behaved.  We show that the leakage of
any embedding for $P_{X,Y}$ is lower bounded by the leakage of some
regular embedding of the same primitive. Thus, in order to lower bound
the leakage of any \correct\ implementation of a given primitive, it
suffices to minimize the leakage over all its regular embeddings.  We
also show that the only non-leaking embeddings are the ones for
trivial primitives, where a primitive $P_{X,Y}$ is said to be {\em
  (cryptographically) trivial} if it can be generated by a classical
protocol against HBC adversaries\footnote{\label{foot:caveat}We are
  aware of the fact that our definition of triviality encompasses
  cryptographically interesting primitives like coin-tossing and
  generalizations thereof for which highly non-trivial protocols
  exist~\cite{Mochon07,CK09}. However, the important fact (for
  the purpose of this paper) is that all these primitives can be
  implemented by \emph{trivial} classical protocols against HBC
  adversaries.}. It follows that any quantum protocol implementing a
non-trivial primitive $P_{X,Y}$ must leak information under the sole
assumption that it produces $(X,Y)$ with the right joint
distribution. This extends known impossibility results for two-party
primitives to all non-trivial primitives.

Embeddings of primitives arise from protocols where Alice and Bob have
full control over the environment. Having in mind that any embedding
of a non-trivial primitive leaks information, it is natural to
investigate what tasks can be implemented without leakage with the
help of a trusted third party. The notion of leakage can easily be
adapted to this scenario.
We show that no cryptographic two-party primitive can be implemented
without leakage with just one call to the ideal functionality of a
weaker primitive\footnote{The weakness of a primitive will be formally
  defined in terms of entropic monotones for classical two-party
  computation introduced by Wolf and Wullschleger~\cite{WW04}, see
  Section~\ref{sec:limitedresources}.}. This new impossibility result
does not follow from the ones known since they all assume that the
state shared between Alice and Bob is pure.

We then turn our attention to the leakage of \correct\ protocols for a
few concrete universal primitives. From the results described above,
the leakage of any \correct\ implementation of a primitive can be
determined by finding the (regular) embedding that minimizes the
leakage.  In general, this is not an easy task since it requires to
find the eigenvalues of the reduced density matrix $\rho_A =
\tr_B\proj{\psi}$ (or equivalently $\rho_B = \tr_A\proj{\psi}$). As
far as we know, no known results allow us to obtain a non-trivial
lower bound on the leakage (which is the difference between the mutual
information and accessible information) of non-trivial primitives. One
reason being that in our setting we need to lower bound this
difference with respect to a measurement in one particular basis.
However, when $P_{X,Y}$ is such that the bit-length of either $X$ or
$Y$ is short, the leakage can be computed precisely. We show that any
\correct\ implementation of \ot\ necessarily leaks $\frac{1}{2}$\,bit.
Since NL-boxes and \ot\ are locally equivalent, the same minimal
leakage applies to NL-boxes~\cite{WW05b}.  This is a stronger
impossibility result than the one by Lo~\cite{Lo97} since he
assumes perfect/statistical privacy against one party while our
approach only assumes \correctness\ (while both approaches apply even
against QHBC adversaries). We finally show that for Rabin-OT and \ot\
of $r$-bit strings (i.e.~\srot\ and \sotr\ respectively), the leakage
approaches $1$ exponentially in $r$. In other words, \correct\
implementations of these two primitives trivialize as $r$ increases
since the sender gets almost all information about Bob's reception of
the string (in case of \srot) and Bob's choice bit (in case of
\sotr). These are the first quantitative impossibility results for
these primitives and the first time the hardness of
implementing different flavors of string OT is shown to increase as
the strings to be transmitted get longer.

Finally, we note that our lower bounds on the leakage of the
randomized primitives also lower-bound the minimum leakage for the
standard versions of these primitives\footnote{The definition of
  leakage of an embedding can be generalized to protocols with inputs,
  where it is defined as $\max\{ \sup_{V_B} S(X;V_B)-I(X;Y) \, , \,
  \sup_{V_A} S(V_A;Y)-I(X;Y) \}$, where $X$ and $Y$ involve both 
inputs and outputs of Alice and Bob, respectively. The supremum is taken over all
  possible (quantum) views $V_A$ and $V_B$ of Alice and Bob obtained
  by their (QHBC-consistent) actions (and containing their inputs).}
where the players choose their
inputs uniformly at random. While we focus on the typical case where
the primitives are run with uniform inputs, the same reasoning can be
applied to primitives with arbitrary distributions of inputs.

\subsection{Related Work} Our framework allows to quantify the
minimum amount of leakage whereas standard impossibility proofs as the
ones of~\cite{LC97,Mayers97,Lo97,AKSW07,BCS12} do not in general
provide such quantification since they usually assume privacy for one
player in order to show that the protocol must be totally insecure for
the other player\footnote{Trade-offs between the security for one and the
  security for the other player have been considered before, but
  either the relaxation of security has to be very small~\cite{Lo97}
  or the trade-offs are restricted to particular primitives such as
  commitments~\cite{SR01,BCHLW08} or oblivious transfer~\cite{CKS13}.}. By contrast, we derive lower
bounds for the leakage of any \correct\ implementation.  At first
glance, our approach seems contradictory with standard impossibility
proofs since embeddings leak the same amount towards both parties.  To
resolve this apparent paradox it suffices to observe that in previous
approaches only the adversary purified its actions whereas in our case
both parties do. If a honest player does not purify his actions then
some leakage may be lost by the act of irreversibly and unnecessarily
measuring some of his quantum registers.

Our results complement the ones obtained by Colbeck
in~\cite{Colbeck07} for the setting where Alice and Bob have inputs
and obtain identical outcomes (called single-function computations).
\cite{Colbeck07} shows that in any \correct\ implementation of
primitives of a certain form, an honest-but-curious player can access
more information about the other party's input than it is available
through the ideal functionality. Unlike~\cite{Colbeck07}, we deal in
our work with the case where Alice and Bob do not have inputs but
might receive different outputs according to a joint probability
distributions. We show that only trivial distributions can be
implemented securely in the QHBC model.  Furthermore, we introduce a
quantitative measure of protocol-insecurity that lets us answer which
embedding allow the least effective cheating.

Another notion of privacy in quantum protocols, generalizing its
classical counterpart from~\cite{BK91,Kushilevitz92}, is proposed by
Klauck in~\cite{Klauck04}.  Therein, two-party quantum protocols with
inputs for computing a function
$f:\mathcal{X}\times\mathcal{Y}\rightarrow \mathcal{Z}$, where
$\mathcal{X}$ and $\mathcal{Y}$ denote Alice's and Bob's respective
input spaces, and privacy against QHBC adversaries are considered.
Privacy of a protocol is measured in terms of \emph{privacy loss},
defined for each round of the protocol and fixed distribution of
inputs $P_{X',Y'}$ by $S(B;X|Y)=H(X|Y)-S(X|B,Y)$, where $B$ denotes
Bob's private working register, and $X\assign(X',f(X',Y'))$,
$Y\assign(Y',f(X',Y'))$ represent the complete views of Alice and Bob,
respectively.  Privacy loss of the entire protocol is then defined as
the supremum over all joint input distributions, protocol rounds, and
states of working registers.  In our framework, privacy loss
corresponds to $S(X;YB)-I(X;Y)$ from Alice point's of view and
$S(Y;XA)-I(X;Y)$ from Bob's point of view.  Privacy loss is therefore
very similar to our definition of leakage except that it requires the
players to get their respective honest outputs.  As a consequence, the
protocol implementing $P_{X,Y}$ by asking one party to prepare a
regular embedding of $P_{X,Y}$ before sending her register to the
other party would have no privacy loss. Moreover, the scenario
analyzed in \cite{Klauck04} is restricted to primitives which provide
the same output $f(X,Y)$ to both players.  Another difference is that
since privacy loss is computed over all rounds of a protocol, a party
is allowed to abort which is not considered QHBC in our setting.  In
conclusion, the model of~\cite{Klauck04} is different from ours even
though the measures of privacy loss and leakage are similar.
\cite{Klauck04} provides interesting results concerning trade-offs
between privacy loss and communication complexity of quantum
protocols, building upon similar results of~\cite{BK91,Kushilevitz92}
in the classical scenario. It would be interesting to know whether a
similar operational meaning can also be assigned to the new measure of
privacy, introduced in this paper. 

A result by K\"unzler et al.~\cite{KMR09} shows that two-party
functions that are securely computable against active quantum
adversaries form a strict subset of the set of functions which are
securely computable in the classical HBC model. This complements our
result that the sets of securely computable functions in both HBC and
QHBC models are the same.

A recent paper by Fehr, Katz, Song, Zhou and Zikas~\cite{FKSZZ13} studies our question with respect to the stricter requirements of universal composability. They give classification results for quantum protocols achieving classical primitives with computational and information-theoretic security. Interestingly, classical and quantum protocols seem to be similarly powerful with respect to computational security whereas in the information-theoretic setting, the two landscapes look different.

\subsection{Roadmap}  In Section~\ref{chap:prelim}, we introduce the
cryptographic and information-theoretic notions and concepts used
throughout the paper. We define, motivate, and analyze the generality
of modeling two-party quantum protocols by embeddings in
Section~\ref{qembprotocols} and define triviality of primitives and
embeddings. In Section~\ref{crypto}, we define the notion of leakage
of embeddings, show basic properties and argue that it is a reasonable
measure of privacy. In Section~\ref{sec:primleakage}, we explicitly
lower bound the leakage of some universal two-party
primitives. Finally, in Section~\ref{conclusion} we discuss possible
directions for future research and open questions.

\section{Preliminaries}
\label{chap:prelim}
\subsection{Quantum Information Theory}  For $x,y\in\{0,1\}^n$, 
$\delta_{x,x}=1$ and $\delta_{x,y}=0$ if $x\neq y$. In the following, 
we denote by $\unit{A}$ the set of unitary transforms acting in Hilbert 
space $\hil_{A}$. Let $\ket{\psi}_{AB} \in
\hil_{AB}$ be an arbitrary pure state of the joint systems $A$ and
$B$. The states of these subsystems are $\rho_A = \tr_B\proj{\psi}$
and $\rho_B=\tr_A\proj{\psi}$, respectively.  We denote by $S(A)_{\psi}
\assign S(\rho_A)$ and $S(B)_{\psi} \assign S(\rho_B)$ the von Neumann
entropy (defined as the Shannon entropy of the eigenvalues of the
density matrix) of subsystem $A$ and $B$ respectively. Whenever the quantum
state $\ket{\psi}$ is clear from the context, we omit the
subscripts from entropic quantities and simply write $S(A)$ and $S(B)$. Since the joint
system is in a pure state, it follows from the Schmidt decomposition
that $S(A)=S(B)$ (see e.g.~\cite{NC00}). Analogously to their
classical counterparts, we can define \emph{quantum conditional entropy}
$S(A|B)\assign S(AB)-S(B)$, and \emph{quantum mutual information}
$S(A;B)\assign S(A)+S(B)-S(AB)=S(A)-S(A|B)=S(B)-S(B|A).$ Note that applying a local unitary transform $U = \id_A \otimes U_B$ to the bipartite state $\rho_{AB}$ does not change the mutual information $S(A;B)_\rho = S(A;B)_{U \rho U^\dag}$, because the spectra of eigenvalues of $\rho_A$, $\rho_B$ and $\rho_{AB}$ remain the same. Even though 
$S(A|B)$ can be negative in general, $S(A|B)\geq 0$ is always true if $A$ is a
classical register.  

Let $R=\{(P_X(x),\rho_R^x\}_{x\in {\cal X}}$ be
an ensemble of states $\rho_R^x$ with prior probability $P_X(x)$. This
defines a classical-quantum (cq) state $\rho_{XR}$ where the
average quantum state is $\rho_R = \sum_{x\in {\cal X}} P_X(x)
\rho_R^x$. The following lemma states that applying a separate unitary
transform to each $\rho_{R}^x$ does not change the entropies $S(XR)$ and
$H(X)$, but it might change $S(R)$.
\begin{lemma} \label{lem:unitaryentropy}
Let $\rho_{XR}=\sum_{x\in {\cal X}} P_X(x)
\rho_R^x$ be a cq-state and let $U_{XR} =\sum_{x} \proj{x}_X \otimes U^{x}_{R}$ be a
unitary transform acting only on register $R$,
conditioned on the classical value $x$ in $X$. Then,
$S(XR)_{\rho_{XR}} = S(XR)_{U_{XR} \rho_{XR} U_{XR}^\dag}$ and
$H(X)_{\rho_{XR}} = H(X)_{U_{XR} \rho_{XR} U_{XR}^\dag}$.
\end{lemma}
\begin{proof}
The density matrix of the cq-state $\rho_{XR}$ is block-diagonal and
applying separate unitary transforms $U^x_R$ in every sub-block does
not change the overall spectrum of eigenvalues. Hence, the entropy
$S(XR)$ remains the same. The second equality follows from the fact
that the unitary $U_{XR}$ only acts on register $R$.\qed
\end{proof}

 The famous result by Holevo upper-bounds the amount of
classical information about $X$ that can be obtained by measuring
$\rho_R$:
\begin{theorem}[Holevo bound~\cite{Holevo73,Ruskai02}]\label{holevo}
  Let $Y$ be the random variable describing the outcome of some
  measurement applied to $\rho_R$ for $R=\{P_X(x),\rho_R^x\}_{x\in
    {\cal X}}$. Then, $ I(X;Y) \leq S(\rho_R)-\sum_xP_X(x)S(\rho_R^x),$
  where equality can be achieved if and only if
  $\{\rho_R^x\}_{x\in{\cal X}}$ are simultaneously diagonalizable.
\end{theorem}
Note that if all states in the ensemble are pure and all different
then in order to achieve equality in the theorem above, they have to
form an orthonormal basis of the space they span. In this case, the
variable $Y$ achieving equality is the measurement outcome in this
orthonormal basis.

\subsection{Markov Chains} We say that three classical random variables $X,Y,Z$ with joint distribution $P_{XYZ}$ form a \emph{Markov chain} $X \leftrightarrow Y \leftrightarrow Z$, if $X$ and $Z$ are independent given $Y$, i.e.,~$P_{XZ|Y}=P_{X|Y} \cdot P_{Z|Y}$. Equivalent conditions are $P_{X|YZ} = P_{X|Y}$ or $P_{Z|YX} = P_{Z|Y}$~\cite{CT91}.
Markov chains with quantum ends have been defined in~\cite{DFSS07} and used in subsequent works such
as~\cite{FS09}. For a ccq-state $\rho_{XYR} = \sum_{x,y}P_{X Y}(x,y)\proj{x} \otimes \proj{y} \otimes \rho_{R}^{x,y}$, we say that $X$, $Y$, $R$ form a Markov chain $X \leftrightarrow Y \leftrightarrow R$, if $\rho_{XYR} = \sum_{x,y}P_{X Y}(x,y)\proj{x} \otimes \proj{y} \otimes \rho_{R}^y$, i.e., the quantum register $R$ depends only on the classical variable $y$ but not on $x$. 
\begin{lemma} \label{lem:qmarkov}
For a ccq-state $\rho_{XYR}$, the following conditions are equivalent:
\begin{enumerate}
\item $X \leftrightarrow Y \leftrightarrow R$
\item $S(X|YR) = S(X|Y)$
\item $S(R|YX) = S(R|Y)$
\item $S(X;YR) = I(X;Y)$\enspace .
\end{enumerate}
\end{lemma}
\begin{proof}
For fixed $x,y$, we can diagonalize $\rho_R^{x,y} = \sum_k \lambda_k^{x,y} \proj{\varphi_k^{x,y}}_R$. By redefining the random variable $Y$ to be $(YK)$ with joint distribution $P_{X(YK)}(x,yk) = P_{XY}(x,y) \lambda^{x,y}_k$, we can assume without loss of generality that $\rho_R^{x,y}=\proj{\varphi^{x,y}}_R$ is a pure state for every fixed $x,y$. 
In that case, it is easy to check that $X \leftrightarrow Y \leftrightarrow R$ implies the other three conditions, because $S(XYR)=S(XY)$ and $S(YR)=S(Y)$. 

On the other hand, if $X \leftrightarrow Y \leftrightarrow R$ does not hold, there exist $x\neq x'$ and $y$ such that $\rho^{x,y}_R \neq \rho^{x',y}_R$. Hence, there exists a measurement on registers $YR$ that reveals more information about $X$ than just knowing $Y$, which implies $S(X|YR) \neq S(X|Y)$. The other implications can be shown similarly.
\qed
\end{proof}

\subsection{Dependent Part}  The following definition introduces a
random variable describing the correlation between two random
variables $X$ and $Y$, obtained by collapsing all values $x_1$ and
$x_2$ for which $Y$ has the same conditional distribution, to a
single value.
\begin{definition}[Dependent part~\cite{WW04}]
\label{dep_part}
For two random variables $X,Y$, let $f_X(x) \assign P_{Y|X=x}$. Then the
\emph{dependent part of $X$ with respect to $Y$} is defined as $\dep{X}{Y} \assign f_X(X)$.
\end{definition}
The dependent part $\dep{X}{Y}$ is the minimum random variable among
the random variables computable from $X$ for which $X\leftrightarrow
\dep{X}{Y} \leftrightarrow Y$ forms a Markov chain \cite{WW04}. In
other words, for any random variable $K=f(X)$ such that
$X\leftrightarrow K \leftrightarrow Y$ is a Markov chain, there exists
a function $g$ such that $g(K)=\dep{X}{Y}$.  Immediately from the
definition we get several other properties of
$\dep{X}{Y}$~\cite{WW04}: $H(Y|\dep{X}{Y})=H(Y|X)$,
$I(X;Y)=I(\dep{X}{Y};Y)$, and $\dep{X}{Y}=\dep{X}{(\dep{Y}{X})}$. The
second and the third formula yield
$I(X;Y)=I(\dep{X}{Y};\dep{Y}{X})$. For two random variables $X$ and
$Z$, we write $X \equiv Z$ if $X$ and $Z$ have the same distributions
(over possibly different alphabets). In particular, we write $X \equiv
\dep{X}{Y}$ if the random variable $X$ consists only of the dependent
part $\dep{X}{Y}$ with respect to $Y$.

The notion of dependent part has been further investigated in
\cite{FWW04,IMNW04,WW05a}.  Wullschleger and Wolf have shown that
quantities $H(X\searrow Y|Y)$ and $H(Y\searrow X|X)$ are monotones for
two-party computation~\cite{WW05a}. That is, none of these values can
increase during classical two-party protocols. In particular, if Alice
and Bob start a protocol from scratch then classical two-party
protocols can only produce $(X,Y)$ such that: $H(X\searrow
Y|Y)=H(Y\searrow X|X)=0$, since $H(X\searrow Y|Y)>0$ if and only if
$H(Y\searrow X|X)>0$~\cite{WW05a}.  Conversely, any primitive
satisfying $H(X\searrow Y|Y)=H(Y\searrow X|X)=0$ can be implemented
securely in the honest-but-curious (HBC) model. We call such
primitives \emph{trivial}\footnote{See Footnote~\ref{foot:caveat} for
  a caveat about this terminology.}.

\subsection{Connected Components}\label{connected}
Another property of a joint probability distribution $P_{XY}$ which we require is the notion of \emph{connected components}, as in~\cite[Def.~1]{WW04}.
 
\begin{definition} \label{def:connected}
Let $X$ and $Y$ be random variables with (disjoint) ranges $\mathcal{X}$ and $\mathcal{Y}$, distributed according to $P_{XY}$. Consider the bipartite graph $G$ with vertex set $\mathcal{X} \cup \mathcal{Y}$ such that two vertices $x \in \mathcal{X}$ and $y \in \mathcal{Y}$ are connected by an edge iff $P_{XY}(x,y) > 0$ holds. We call the edge sets $\mathcal{C}_1, \ldots \mathcal{C}_\ell$ of connected components of the graph $G$ the \emph{connected components of $P_{XY}$}.
\end{definition}

In this way, the joint distribution $P_{XY}$ can be split into $\ell$ distributions $\{P_{X_j,Y_j}\}_{j=1}^\ell$. For every $j$, $P_{X_j,Y_j}$ is a distribution with a single component over alphabet $\mathcal{X}_j \times \mathcal{Y}_j$, where 
$\mathcal{X}$ is the disjoint union of the $\mathcal{X}_j$ and
$\mathcal{Y}$ the disjoint union of the $\mathcal{Y}_j$. We denote by
the random variable $C$ the component of $XY$, resulting in the joint distribution
$P_{CXY}$. Then, $P_C(j) = \sum_{xy \in \mathcal{C}_j} P_{XY}(x,y) = \sum_{x \in \mathcal{X}_j} P_X(x) = \Pr(X
\in \mathcal{X}_j) = \sum_{y \in \mathcal{Y}_j} P_Y(y) = 
\Pr(Y \in \mathcal{Y}_j)$ is the probability that $XY$ ends up in
component $\mathcal{C}_j$ (which is the same as the probability that $X$ ends up
in $\mathcal{X}_j$ and that $Y$ ends up in $\mathcal{Y}_j$).
Note that $C$ is a deterministic function of $X$ (and also of $Y$), hence 
\begin{align} 
\begin{split} \label{eq:component}
I(X;Y) &= H(Y) - H(Y|X) = H(YC) - H(Y|XC) = H(C) + H(Y|C) - H(Y|XC)\\ 
&= H(C) + I(X;Y|C)\enspace .
\end{split}
\end{align}

\subsection{Purification}
\label{purification}
All security questions we ask are with respect to \emph{(quantum)
  honest-but-curious} adversaries. In the classical honest-but-curious
adversary model (HBC), the parties follow the instructions of a
protocol but store all information available to them. Quantum
honest-but-curious adversaries (QHBC), on the other hand, are allowed
to behave in an arbitrary way that cannot be distinguished from their
honest behavior by the other player.

Almost all impossibility results in quantum cryptography rely upon a
quantum honest-but-curious behavior of the adversary.  This behavior
consists in {\em purifying} all actions of the honest
players. Purifying means that instead of invoking classical randomness
from a random tape, for instance, the adversary relies upon quantum
registers holding all random bits needed. The operations to be
executed from the random outcome are then performed quantumly without
fixing the random outcomes.  For example, suppose a protocol instructs
a party to pick with probability $p$ state $\ket{\phi^0}_C$ and with
probability $1-p$ state $\ket{\phi^1}_C$ before sending it to the
other party through the quantum channel $C$. The purified version of
this instruction looks as follows: Prepare a quantum register in state
$\sqrt{p}\ket{0}_R+\sqrt{1-p}\ket{1}_R$ holding the random
process. Add a new register initially in state $\ket{0}_C$ before
applying the unitary transform $U:\ket{r}_R\ket{0}_C \mapsto
\ket{r}_R\ket{\phi^r}_C$ for $r\in\{0,1\}$, send register $C$
through the quantum channel and keep register $R$.

>From the receiver's point of view, the purified behavior is
indistinguishable from the one relying upon a classical source of
randomness because in both cases, the state of register $C$ is
$\rho=p\proj{\phi^0}+(1-p)\proj{\phi^1}$. All operations invoking
classical randomness can be purified
similarly~\cite{LC97,Mayers97,Lo97,Kent04}. The result is that
measurements are postponed as much as possible and only extract
information required to run the protocol in the sense that only when
both players need to know a random outcome, the corresponding quantum
register holding the random coin will be measured.  If both players
purify their actions then the joint state at any point during the
execution will remain pure, until the very last step of the protocol
when the outcomes are measured.

\subsection{Secure Two-Party Computation}
\label{prel.crypto}
In Section~\ref{sec:primleakage}, we investigate the leakage of
several universal cryptographic two-party primitives. By universality
we mean that any two-party secure function evaluation can be reduced
to them. We investigate the completely randomized versions where
players do not have inputs but receive randomized outputs instead.
Throughout this paper, the term \emph{primitive} usually refers to the
joint probability distribution defining its randomized version. Any
protocol implementing the standard version of a primitive (with
inputs) can also be used to implement a randomized version of the same
primitive, with the ``inputs'' chosen according to an arbitrary fixed
probability distribution.

\section{Two-Party Protocols and Their Embeddings}\label{qembprotocols}

\subsection{\CorrectNess}\label{protocols}
In this work, we consider \emph{cryptographic primitives} providing
$X$ to honest player Alice and $Y$ to honest player Bob according to a
joint probability distribution $P_{X,Y}$. The goal of this section is
to define when a protocol $\pi$ \emph{correctly implements} the
primitive $P_{X,Y}$. The first natural requirement is that once the
actions of $\pi$ are purified by both players, measurements of
registers $A$ and $B$ in the computational basis\footnote{It is clear
  that every quantum protocol for which the final measurement
  (providing $(x,y)$ with distribution $P_{X,Y}$ to the players) is
  not in the computational basis can be transformed into a protocol of
  the described form by two additional local unitary transformations.}
provide joint outcome $(X,Y)=(x,y)$ with probability $P_{X,Y}(x,y)$.

Protocol $\pi$ can use extra registers $A'$ on Alice's and $B'$ on
Bob's side providing them with (quantum) working space.  The
purification of all actions of $\pi$ therefore generates a pure state
$\ket{\psi}\in \hil_{AB}\otimes \hil_{A'B'}$. A second requirement for
the correctness of the protocol $\pi$ is that these extra registers
are only used as working space, i.e.~the final state
$\ket{\psi}_{ABA'B'}$ is such that the content of Alice's working
register $A'$ does not give her any further information about Bob's
output $Y$ than what she can infer from her honest output $X$ and vice
versa for $B'$. Formally, we require that $S(XA' ; Y)=I(X;Y)$ and $S(X
; YB')=I(X;Y)$. By Lemma~\ref{lem:qmarkov}, the two conditions are equivalent to requiring $A' \leftrightarrow X \leftrightarrow Y \leftrightarrow B'$ to be a Markov chain.

\begin{definition}\label{defcorrect}
\label{correct}
A protocol $\pi$ for $P_{X,Y}$ is \emph{\correct}\ if measuring
registers $A$ and $B$ of its final state in the computational basis
yields outcomes $X$ and $Y$ with distribution $P_{X,Y}$ and the final
state satisfies $S(X; YB') = S(XA'; Y) =I(X;Y)$ where $A'$ and $B'$
denote the extra working registers of Alice and Bob. The state
$\ket{\psi}\in \hil_{AB} \otimes \hil_{A'B'}$ is called an
\emph{embedding of $P_{X,Y}$} if it can be produced by the
purification of a \correct\ protocol for $P_{X,Y}$.
\end{definition}
We would like to point out that our definition of correctness is
stronger than the usual classical notion which only requires the
correct distribution for the output of the honest players. For example,
the trivial classical protocol for the primitive $P_{X,Y}$ in which
Alice samples both player's outputs $XY$, sends $Y$ to Bob, but keeps
a copy of $Y$ for herself, is not \emph{\correct}\ because it 
implements a fundamentally different primitive,
namely $P_{XY,Y}$. Definition~\ref{defcorrect} requires that any protocol
for $P_{X,Y}$  leaks no information beyond $I(X;Y)$
to any party having measured its output $X$ or $Y$.

\subsection{Regular Embeddings} \label{sec:embed} We call an embedding
$\ket{\psi}_{ABA'B'}$ \emph{regular} if the working registers $A',B'$
are empty. Formally, let ${\Theta}_{n,m} \assign
\{\theta:\{0,1\}^n\times\{0,1\}^m \rightarrow [0\ldots2\pi) \}$ be the
set of functions mapping bit-strings of length $m+n$ to real numbers
between $0$ and $2\pi$.

\begin{definition}
For a joint probability distribution $P_{X,Y}$ where $X\in\{0,1\}^n$
and $Y\in\{0,1\}^m$, we define the set
\[ \emb{P_{X,Y}} \assign \left\{\ket{\psi} \in \hil_{AB} : \ket{\psi}
  = \!\!\!\!\!\!\!\! \sum_{x\in\{0,1\}^n \! , \, y \in\{0,1\}^m}
  \!\!\!\!\!\!\!\!
  \expe{i\theta(x,y)}\sqrt{P_{X,Y}(x,y)}\ket{x,y}_{AB} \, , \theta\in
  {\Theta}_{n,m} \right\} \, ,
\]
and call any state $\ket{\psi}\in \emb{P_{X,Y}}$ a {\em regular
  embedding} of the joint probability distribution $P_{X,Y}$ .
\end{definition}

Clearly, any $\ket{\psi}\in \emb{P_{X,Y}}$ produces $(X,Y)$ with
distribution $P_{X,Y}$ since the probability that Alice measures $x$
and Bob measures $y$ in the computational basis is
$|\bracket{\psi}{x,y}|^2= P_{X,Y}(x,y)$. In order to specify a
particular regular embedding one only needs to give the description of
the {\em phase function} $\theta(x,y)$. We denote by
$\ket{\psi_\theta}\in \emb{P_{X,Y}}$ the quantum embedding of
$P_{X,Y}$ with phase function~$\theta$.  The constant function
$\theta(x,y) \assign 0$ for all $x\in\{0,1\}^n, y\in\{0,1\}^m$
corresponds to what we call {\em canonical embedding}
$\ket{\psi_{\vec{0}}} \assign
\sum_{x,y}\sqrt{P_{X,Y}(x,y)}\ket{x,y}_{AB}$.

In Lemma~\ref{super_leak} below we show that every primitive $P_{X,Y}$
has a regular embedding which is in some sense the most secure among
all embeddings of $P_{X,Y}$.

\subsection{Trivial Classical Primitives and Trivial Embeddings}
In this section, we define \emph{triviality} of classical primitives
and (bipartite) embeddings. We show that for any non-trivial classical
primitive, its canonical quantum embedding is also non-trivial.
Intuitively, a primitive $P_{X,Y}$ is {\em trivial} if $X$ and $Y$ can
be generated by Alice and Bob from scratch in the classical
honest-but-curious (HBC) model\footnote{See Footnote~\ref{foot:caveat}
  for a caveat about this terminology.}. Formally, we define
triviality via an entropic quantity based on the notion of
\emph{dependent part} (see Section~\ref{chap:prelim}).
\begin{definition}
  A primitive $P_{X,Y}$ is called {\em trivial} if it satisfies
  $H(\dep{X}{Y}|Y)=0$, or equivalently,
  \mbox{$H(\dep{Y}{X}|X)=0$}. Otherwise, the primitive is called {\em
    non-trivial}.
\end{definition}

\begin{definition}
  A regular embedding $\ket{\psi}_{AB}\in \emb{P_{X,Y}}$ is called
  {\em trivial} if either $S(\dep{X}{Y}|B)=0$ or
  $S(\dep{Y}{X}|A)=0$. Otherwise, we say that $\ket{\psi}_{AB}$ is
  {\em non-trivial}.
\end{definition}
Notice that unlike in the classical case,
$S(\dep{X}{Y}|B)=0\Leftrightarrow S(\dep{Y}{X}|A)=0$ does not hold in
general. As an example, consider a shared quantum state where the
computational basis corresponds to the Schmidt basis for only one of
its subsystems, say for $A$. Let
$\ket{\psi}=\alpha\ket{0}_A\ket{\xi_0}_B+\beta\ket{1}_A\ket{\xi_1}_B$
be such that both subsystems are two-dimensional,
$\{\ket{\xi_0},\ket{\xi_1}\}\neq \{\ket{0},\ket{1}\}$,
$\bracket{\xi_0}{\xi_1}=0$, and $|\bracket{\xi_0}{0}|\neq
|\bracket{\xi_1}{0}|$.  We then have $S(X|B)=0$ and $S(Y|A)>0$ while
$X\equiv \dep{X}{Y}$ and $Y\equiv \dep{Y}{X}$.

To illustrate this definition of triviality, we argue in the following
that if a primitive $P_{X,Y}$ has a trivial regular embedding, there
exists a classical protocol which generates $X,Y$ securely in the HBC
model.  Let $\ket{\psi}\in \emb{P_{X,Y}}$ be trivial and assume
without loss of generality that $S(\dep{Y}{X}|A)=0$.  Intuitively,
this means that Alice can learn everything possible about Bob's
outcome $Y$ ($Y$ could include some private coin-flips on Bob's side,
but that is ``filtered out'' by the dependent part). More precisely,
Alice holding register $A$ can measure her part of the shared state to
completely learn a realization of $\dep{Y}{X}$, specifying
$P_{X|Y=y}$. 
She then chooses $X$ according to the distribution $P_{X|Y=y}$.  An
equivalent way of trivially generating $(X,Y)$ classically is the
following classical protocol:
\begin{enumerate}
\item Alice samples $y'$ from distribution $P_{\dep{Y}{X}}$ and
  announces the outcome to Bob.
\item Alice samples $x$ from distribution $P_{X|\dep{Y}{X}=y'}$.
\item Bob samples $y$ from distribution $P_{Y|\dep{Y}{X}=y'}$.
\end{enumerate}
Of course, the same reasoning applies in case $S(\dep{X}{Y}|B)=0$ with
the roles of Alice and Bob reversed. 

In fact, the following lemma
shows that any non-trivial primitive
$P_{X,Y}$ has a non-trivial embedding, i.e.~there exists a quantum
protocol \correctly\ implementing $P_{X,Y}$ while leaking less
information to QHBC adversaries than any classical protocol for
$P_{X,Y}$ in the HBC model.
\begin{lemma}
\label{nontrivialemb}
If $P_{X,Y}$ is a non-trivial primitive then the canonical embedding
$\ket{\psi_{\vec{0}}}\in \emb{P_{X,Y}}$ is also non-trivial.
\end{lemma}
\begin{proof}
A non-trivial embedding of $P_{X,Y}$ can be created from a non-trivial embedding of  $P_{\dep{X}{Y},\dep{Y}{X}}$ by applying local unitary transforms.  We therefore assume
without loss of generality that $X\equiv \dep{X}{Y}$ and $Y\equiv \dep{Y}{X}$.
Let
$$\ket{\psi_{\vec{0}}} \assign \sum_{x,y}\sqrt{P_{X,Y}(x,y)}\ket{x,y}$$ be the canonical
embedding of $P_{X,Y}$.
Since $X\equiv \dep{X}{Y}$ and $Y\equiv \dep{Y}{X}$, it holds for any
$x_0\neq x_1$ that $P_{Y|X=x_0}\neq P_{Y|X=x_1}$. Furthermore, since $P_{X,Y}$
is non-trivial, there exist $x_0\neq x_1$ and $y_0$ such that $P_{Y|X=x_0}(y_0)>0$
and $P_{Y|X=x_1}(y_0)>0$.
The state $\ket{\psi_{\vec{0}}}$ can be written in the form:
\begin{eqnarray*}
\ket{\psi_{\vec{0}}}&=&\sqrt{P_{X}(x_0)}\ket{x_0}\sum_{y}\sqrt{P_{Y|X=x_0}(y)}\ket{y}+\sqrt{P_{X}(x_1)}\ket{x_1}\sum_{y}\sqrt{P_{Y|X=x_1}(y)}\ket{y}+\ket{\psi'}\enspace ,
\end{eqnarray*}
where
$\tr(\proj{x_0}\tr_B\proj{\psi'})=\tr(\proj{x_1}\tr_B\proj{\psi'})=0$.
Set $\ket{\varphi^{x_b}}\assign \sum_y \sqrt{P_{Y|X=x_b}(y)}\ket{y}$ for $b\in\{0,1\}$.
Since $P_{Y|X=x_0}\neq P_{Y|X=x_1}$, we get that
$|\bracket{\varphi^{x_0}}{\varphi^{x_1}}|<1$.  Because all coefficients
at $\ket{y}$ in the normalized vectors $\ket{\varphi^{x_0}}$ and
$\ket{\varphi^{x_1}}$  are non-negative, and the coefficients at
$\ket{y_0}$ are both positive,
$\bracket{\varphi^{x_0}}{\varphi^{x_1}}\neq 0$. Therefore, the
non-identical states $\ket{\varphi^{x_0}}$ and $\ket{\varphi^{x_1}}$
cannot be perfectly distinguished, which implies that Bob cannot learn
whether $X=x_0$ or $X=x_1$ with probability 1. Therefore, the von
Neumann entropy on Bob's side $S(B)$ is such that $S(B)<H(X)$. As shown in~\cite{WW05a}, 
$H(\dep{X}{Y}|Y)>0$ implies $H(\dep{Y}{X}|X)>0$, and we can argue in the
same way as above that $S(A)<H(Y)$ from which follows that
$\ket{\psi_{\vec{0}}}$ is a non-trivial quantum embedding of
$P_{X,Y}$.  \qed
\end{proof}

\section{The Leakage of Quantum Embeddings}\label{crypto}

In this section, we formally define the leakage of embeddings and establish properties
of the leakage. 

\subsection{Definition and Basic Properties of Leakage}\label{sleakage}
A perfect implementation of $P_{X,Y}$ simply provides $X$ to Alice and
$Y$ to Bob and does nothing else. The expected amount of information
that one random variable gives about the other is $I(X;Y)=
H(X)-H(X|Y)=H(Y)-H(Y|X) = I(Y;X)$. Intuitively, we define the {\em
  leakage of a quantum embedding $\ket{\psi}_{ABA'B'}$ of $P_{X,Y}$}
as the larger of the two following quantities: the extra amount of
information Bob's quantum registers $BB'$ provide about $X$ and the
extra amount Alice's quantum state in $AA'$ provides about $Y$
respectively in comparison to ``the minimum amount''
$I(X;Y)$.\footnote{\label{foot:guess}There are other natural
  candidates for the notion of leakage such as the difference in
  difficulty between guessing Alice's output $X$ by measuring Bob's
  final quantum state $B$ and based on the output of the ideal
  functionality $Y$. While such definitions do make sense, they turn
  out not to be as easy to work with and it is an open question
  whether the natural properties described later in this section can
  be established for these notions of leakage as well.}

\begin{definition}
Let $\ket{\psi}\in \hil_{ABA'B'}$ be an embedding of $P_{X,Y}$. 
We define the leakage $\ket{\psi}$ as
$$\Delta_{\psi}(P_{X,Y}):=\max \left\{ S(X;BB')-I(X;Y) \, , \, S(AA';Y)-I(X;Y) \right\}\, .$$
Furthermore, we say that $\ket{\psi}$ is 
{\em $\delta$-leaking} if $\Delta_{\psi}(P_{X,Y})\geq\delta$ .
\end{definition}

It is easy to see that the leakage is non-negative since $S(X;BB')\geq
S(X;\tilde{B})$ for $\tilde{B}$ the result of a quantum operation
applied to $BB'$.  Such an operation could be the trace over the extra
working register $B'$ and a measurement in the computational basis of
each qubit of the part encoding $Y$, yielding
$S(X;\tilde{B})=I(X;Y)$.

We want to argue that our notion of leakage is a good measure for the
privacy of the player's outputs. In the same spirit, we will argue
that the minimum achievable leakage for a primitive is related to the
``hardness'' of implementing it. We start off by proving several basic
properties about leakage.

For a general state in $\hil_{ABA'B'}$ the quantities
$S(X;BB')-I(X;Y)$ and $S(AA';Y)-I(X;Y)$ are not necessarily
equal. Note though that they coincide for regular embeddings
$\ket{\psi}\in\emb{P_{X,Y}}$ produced by a \correct\ protocol (where the
work spaces $A'$ and $B'$ are empty): Notice that $ S(X;B) =
S(X)+S(B)-S(X,B) = H(X)+S(B) - H(X) = S(B)$ and because $\ket{\psi}$
is pure, $S(A)=S(B)$. Therefore, $S(X;B)=S(A;Y)$ and the two
quantities coincide. The following lemma states that this actually
happens for \emph{all} embeddings and hence, the definition
of leakage is symmetric with respect to both players.

\begin{lemma}[Symmetry]\label{symmth}
\label{symmetry}
Let $\ket{\psi}\in \hil_{ABA'B'}$ be an embedding of
$P_{X,Y}$. Then,
$$\Delta_\psi(P_{X,Y})=S(X;BB')-I(X;Y)=S(AA';Y)-I(X;Y)\enspace .$$
\end{lemma}
\begin{proof}
We have already shown that the statement is true in the case where both $A'$ 
and $B'$ are trivial. In the case where $A'$ is trivial and
$B'$ is not, the Markov chain condition 
implies that $\ket{\psi}$ is of the form
$$\ket{\psi}=\sum_{x,y}\sqrt{P_{X,Y}(x,y)}\ket{x,y}_{AB}\ket{\varphi^y}_{B'}\enspace ,$$
hence, Bob can fix $y_0$ and apply a unitary transform $U_{BB'}$
on his part of the system,
such that $U_{BB'}\ket{y,\varphi^y}=\ket{y,\varphi^{y_0}}$,
and
$$\id_A\otimes U_{BB'}\ket{\psi}_{ABB'}=\ket{\psi^*}_{AB}\otimes \ket{\varphi^{y_0}}_{B'}\enspace ,$$
where $\ket{\psi^*}\in\emb{P_{X,Y}}$. Note that the unitary transform
$U_{BB'}$ does not change the entropic quantity $S(X;BB')_{\ket{\psi}}
= S(X;BB')_{U_{BB'} \ket{\psi}}$. Hence, in the resulting product
state,  we have that
$S(X;BB')-I(X;Y)=S(X;B)-I(X;Y)=S(A;Y)-I(X;Y)$, due to the fact that 
$\ket{\psi^*}\in\emb{P_{X,Y}}$. An analogous statement holds
in the case where $B'$ is trivial and $A'$ is non-trivial. 

We now assume that both $A'$ and $B'$ are non-trivial. 
An embedding of $P_{X,Y}$ can be written as 
\begin{align*} 
\ket{\psi} &= \sum_{x,y} \sqrt{P_{X,Y}(x,y)} \ket{x,y}_{AB}
\ket{\varphi^{x,y}}_{A'B'}\\
&= \sum_j \sqrt{P_C(j)} \sum_{x \in \mathcal{X}_j,y \in \mathcal{Y}_j} \sqrt{P_{X,Y|C=j}(x,y)} \ket{x,y}_{AB}
\ket{\varphi^{x,y}}_{A'B'} \\
&= \sum_j \sqrt{P_C(j)} \ket{\psi_j}_{ABA'B'} \enspace ,
\end{align*}
where $C$ denotes the connected component of $X,Y$ (see Section~\ref{connected}) and where for any $j$, $\ket{\psi_j}$ is an embedding of the single-component primitive $P_{X_j,Y_j}$.

We want to show that 
\begin{equation} \label{eq:symmetry} 
S(X;BB')_{\psi} - I(X;Y) = S(AA';Y)_{\psi} - I(X;Y) \enspace .
\end{equation} 
Using the reasoning of Equation~\eqref{eq:component} for the three
terms $S(X;BB'), I(X;Y), S(AA';Y)$, Equation~\eqref{eq:symmetry} is
equivalent to\footnote{\label{footnote14}The only step that needs some extra thought is
  the following: $S(X|BB') = S(X|BB'C)$ holds, because the component
  $C$ can be determined with certainty by measuring register $B$ with
  projectors $\{\sum_{y \in \mathcal{Y}_j} \proj{y}_B \}_j$.} 
\begin{equation*} 
S(X;BB'|C)_{\psi} - I(X;Y|C) = S(AA';Y|C)_{\psi} - I(X;Y|C) 
\end{equation*}
and hence, it suffices to show symmetry for all single-component primitives $P_{X_j,Y_j}$ and their embeddings $\ket{\psi_j}$. For the rest of the proof, we drop the index $j$ for the ease of notation.
\medskip

Note that $S(X;BB') = H(X)+S(BB')-S(XBB')$ and
$S(AA';Y)=H(Y)+S(AA')-S(AA'Y)$. As $\ket{\psi}$ is a pure state, we
have that $S(AA')_\psi = S(BB')_\psi$, and it suffices to
show that 
\begin{equation} \label{eq:toshow}
H(X) - S(XBB') _\psi = H(Y) - S(AYA') _\psi \enspace .
\end{equation}

For every $x$ and $y$, we can write the bipartite pure
state 
$$\ket{\varphi^{x,y}}_{A'B'} =
\sum_{k=1}^K \sqrt{\lambda_k^{x,y}} \ket{e_k^{x,y}}_{A'} \ket{f_k^{x,y}}_{B'}$$ 
in Schmidt form. For the reduced density matrices, we obtain
$$\rho_{A'}^{x,y} = \sum_k \lambda_k^{x,y} \proj{e_k^{x,y}}\enspace .$$

Since any embedding $\ket{\psi}\in \hil_{ABA'B'}$ of $P_{X,Y}$ is
produced by a \correct\ protocol, it satisfies
$$S(XA';Y)=S(X;YB')=I(X;Y)$$ 
which is
equivalent by Lemma~\ref{lem:qmarkov} to $A'\leftrightarrow X\leftrightarrow Y$ and
$X\leftrightarrow Y\leftrightarrow B'$ being Markov chains. It follows
that for every $x$ and $y \neq y'$ in the same connected component of $P_{XY}$, the reduced density matrices
$\rho_{A'}^{x,y}=\rho_{A'}^{x,y'}=\rho_{A'}^x$ coincide and therefore,
the eigenvalues $\lambda_k^{x,y}$ cannot depend on $y$. Because of
$X\leftrightarrow Y\leftrightarrow B'$, they can neither depend on
$x$. Hence, $\ket{\varphi^{x,y}} = \sum_k \sqrt{\lambda_k}e^{i\theta'(k,x,y)}
\ket{e_k^{x,y}}\ket{f_k^{x,y}}$.\footnote{We note that it is only possible to draw this conclusion within the same connected component. The eigenvalues $\lambda_k^{x,y}$ and $\lambda_k^{x',y'}$ for $x,y$ and $x',y'$ not in the same connected component of $P_{XY}$ cannot be related to each other.} The phase factors arise from the fact
that from a reduced density matrix the global phases of the
Schmidt-basis elements cannot be determined.

Let us fix a set of orthogonal states $\{\ket{k}\}_k$. We define the
unitary transformation $U_{ABA'B'}$  to map the orthonormal states $\{\ket{e_k^{x,y}}_{A'}\}_k$ into the orthonormal
states $\{\ket{k}_{A'}\}_k$, and $\{\ket{f_k^{x,y}}_{B'}\}_k$ into $\{\ket{k}_{B'}\}_k$.
Note that $U_{ABA'B'}$ only acts on
registers $A'B'$ conditioned on the $x$-value in $A$ and the $y$-value
in $B$. Applying $U_{ABA'B'}$ to $\ket{\psi}$ results into state
\begin{align*}
\ket{\chi} &= \sum_{x,y} \sqrt{P_{X,Y}(x,y)} \ket{x,y}_{AB} \sum_k
\sqrt{\lambda_k} e^{i\theta'(k,x,y)}\ket{k,k}_{A'B'}\\
&= \sum_k \sqrt{\lambda_k} \left(\sum_{x,y}\sqrt{P_{X,Y}(x,y)}
e^{i\theta'(k,x,y)}\ket{x,y}\right) \ket{k,k}\\
&= \sum_k \sqrt{\lambda_k} \ket{\chi_k}_{AB} \otimes \ket{k,k}_{A'B'}\enspace,
\end{align*}  
where each $\ket{\chi_k}_{AB}\in\emb{P_{X,Y}}$. 
The cqq-state $\sigma_{XBB'}$ can now be written in the form:
$$\sigma_{XBB'}=\sum_x P_X(x)\proj{x}\otimes\sum_k \lambda_k
\proj{\gamma^x_k,k} \enspace,$$
where $\ket{\gamma^x_k}=\sum_y \sqrt{P_{Y|X=x}}e^{i\theta'(k,x,y)}\ket{y}$. 
Due to the second register, the states $\ket{\gamma^x_k,k}$ are mutually 
orthogonal for each $x$. Therefore, for each $x$,  
$$S\left(\sum_k\lambda_k \proj{\gamma^x_k,k}\right)=H(\lambda_1,\dots,\lambda_K)\enspace.$$
As a result we get that
$$S(XBB')_\chi=H(X)+\sum_x P_X(x)H(\lambda_1,\dots,\lambda_K)=H(X)+H(\lambda_1,\dots,\lambda_K)$$ 
and analogously, 
$$S(AA'Y)_\chi=H(Y)+H(\lambda_1,\dots,\lambda_K)\enspace .$$
Equation~\eqref{eq:toshow} now follows by applying
Lemma~\ref{lem:unitaryentropy} in the first and last step of the
following equations.
\begin{align*}
H(X) - S(XBB')_\psi &= H(X) - S(XBB')_\chi \\
&= - H(\lambda_1,\dots,\lambda_K)\\
&=  H(Y) - S(AYA')_\chi \\
&=  H(Y) - S(AYA')_\psi \enspace .
\end{align*}
\qed
\end{proof}

If a primitive $P_{X,Y}$ has multiple connected components and $\ket{\psi_j}$ are (not necessarily regular) embeddings of $P_{X_j,Y_j}$, then the state $\ket{\psi} \assign \sum_j \sqrt{P_C(j)} \ket{\psi_j}$ is an embedding of $P_{X,Y}$ with leakage 
\begin{align}
\begin{split} \label{eq:splitleak}
\Delta_\psi(P_{X,Y}) &= S(X;BB')_\psi - I(X;Y) = S(X;BB'|C)_\psi - I(X;Y|C)\\
&= \sum_j P_C(j) \Delta_{\psi_j}(P_{X_j,Y_j}) \enspace,
\end{split}
\end{align}
by the same reasoning as in the previous proof (along the lines of
Equation~\eqref{eq:component}). Any party can determine the active component without
disturbing the state once the other party got his/her
output (see Footnote~\ref{footnote14}). Therefore, measuring the component $C$ can be done without changing the amount of information the state contains about the 
other party's output. Hence, we can always assume that the parties know the current
component in use.

The next lemma shows that the leakage of an embedding for a given 
primitive is always lower-bounded by the leakage of some regular embedding 
of the same primitive, which simplifies the calculation of lower bounds for the leakage of embeddings.

\begin{lemma}
\label{super_leak}
For every embedding $\ket{\psi}$ of a primitive $P_{X,Y}$, there exists
 $\ket{\psi^*}\in \emb{P_{X,Y}}$ such that
$\Delta_\psi(P_{X,Y})\geq \Delta_{\psi^*}(P_{X,Y}).$
\end{lemma}
\begin{proof}
In the case where $A'$ and $B'$ are both trivial, then
$\ket{\psi}\in\emb{P_{X,Y}}$ is a regular embedding and the statement
holds trivially. In the case where $A'$ is trivial and $B'$ is not, we
have shown at the beginning of the proof of Lemma~\ref{symmetry} that an embedding
$\ket{\psi}$ of $P_{X,Y}$ is locally equivalent to a state
$\ket{\psi'}_{AB}\otimes \ket{\sigma}_{B'}$ for $\ket{\psi'}\in\emb{P_{X,Y}}$
and a pure state $\ket{\sigma}_{B'}$.  An analogous statement holds if $B'$
is trivial and $A'$ is not.  Therefore, in these two cases we get for
some $\ket{\psi'}\in\emb{P_{X,Y}}$ that
$\Delta_\psi(P_{X,Y})=\Delta_{\psi'}(P_{X,Y})$.

Now assume that both $A'$ and $B'$ are non-trivial and
that $P_{X,Y}$ has multiple connected components. As in the proof of Lemma~\ref{symmetry},
the state $\ket{\psi}_{ABA'B'}$ can be written as 
\[ \ket{\psi}_{ABA'B'} = \sum_j \sqrt{P_C(j)} \ket{\psi_j}_{ABA'B'} \enspace ,
\]
where $\ket{\psi_j}$ is an embedding of $P_{X_j,Y_j}$, the primitive
corresponding to the $j$th connected component of $P_{X,Y}$. Let us assume for now that the lemma holds for single-component primitives. In that case, we get for every $j$ and embedding $\ket{\psi_j}$ a regular embedding $\ket{\psi^*_j} \in \emb{P_{X_j,Y_j}}$ such that $\Delta_{\psi_j}(P_{X_j,Y_j}) \geq \Delta_{\psi^*_j}(P_{X_j,Y_j})$. We define $\ket{\psi^*} = \sum_j \sqrt{P_C(j)} \ket{\psi^*_j}$ and conclude that 
\begin{align*}
\Delta_\psi(P_{X,Y}) &= \sum_j P_C(j) \Delta_{\psi_j}(P_{X,Y})
\geq \sum_j P_C(j) \Delta_{\psi^*_j}(P_{X,Y}) = \Delta_{\psi^*}(P_{X,Y}) \enspace ,
\end{align*}
where the equalities are due to Equation~\eqref{eq:splitleak}.
\medskip

It remains to show the lemma for single-component primitives $P_{X,Y}$.
The state $\ket{\psi}_{ABA'B'}$ is of the form established in the proof
of Lemma~\ref{symmetry}:
\begin{equation} \label{eq:schmidt}
\ket{\psi}_{ABA'B'} = \sum_{x,y}\sqrt{P_{X,Y}(x,y)}
\ket{x,y}_{AB}\otimes\sum_{k}\sqrt{\lambda_k}
e^{i\theta(k,x,y)}\ket{e^{x,y}_k}_{A'}\ket{f^{x,y}_k}_{B'} \enspace .
\end{equation}
Let $\lambda=(\lambda_1,\lambda_2,
\ldots,\lambda_{t})$ be an ordering of all eigenvalues $\{\lambda_k\}_k$ each repeated 
as many times as their multiplicity.
Let $F_{x,y}=\{f^{x,y}_{k}\}_k$ be the set of eigenvectors in $B'$ for each pair $(x,y)$.
Since $X\leftrightarrow Y \leftrightarrow B'$ is a Markov chain, the eigenvectors $f^{x,y}_k$ can be chosen such that $F_{x,y}=F_{x',y}=: F_y$ for any $x,x',y$ in the same connected component. Let us fix an ordering of the elements of
$F_y$, $\langle F_y \rangle =\langle f^y_1,f^y_2, \ldots, f^y_{t}\rangle$,
such that eigenvector $f^y_h$ has eigenvalue $\lambda_h$ whenever $y\in {\cal Y}$.  
\footnote{The Markov chain condition guarantees that a single ordering $\langle F_y \rangle$
suffices in the following sense: 
two eigenvectors $f^{x,y}_k \in F_{x,y}$ and $f^{x',y}_{k'} \in F_{x',y}$ such that 
$f^{x,y}_k=f^{x',y}_{k'}=f^y_h$ for some $f^y_h \in F_y$ necessarily have the same eigenvalue $\lambda_h$. 
}

Consider the (incomplete) projective measurement ${\cal M}=\{\mathbb{Q}_h\}_h$ 
with measurement operators
\[ \mathbb{Q}_h = \sum_{y\in {\cal Y}} \proj{y}_B \otimes \proj{f^y_h}_{B'} \enspace .
\]
Now, suppose that ${\cal M}$ is applied to registers $BB'$ of $\ket{\psi}_{ABA'B'}$. It is easy to
verify that with probability $\lambda_h$, outcome $h$ will be obtained and
the state will collapse to:
\[ \ket{\psi_h}_{ABA'B'} = \sum_{x,y} \sqrt{P_{X,Y}(x,y)} \ket{x,y}_{AB}\otimes e^{i\theta(k(h,x,y),x,y)} \ket{e^{x,y}_{k(h,x,y)}}_{A'}
\otimes \ket{f^y_h}_{B'}
\enspace ,
\]
where $k(h,x,y)$ is the index such that $\ket{e^{x,y}_{k(h,x,y)}}$ is associated with
$\ket{f^y_h}$ in the Schmidt decomposition~\eqref{eq:schmidt} when $X=x$ and $Y=y$. Notice that $\ket{\psi_h}$ is an embedding of $P_{X,Y}$. 
Let $U_h \in \unit{BB'}$ be the local unitary transform on $BB'$ defined as:
\[ U_h \ket{y}_B \ket{f^y_h}_{B'} = \ket{y}_B\ket{\mathbf{0}}_{B'}\enspace ,
\]
and let $\ket{\widehat{\psi}_h}=(\I_{AA'}\otimes U_h)\ket{\psi_h}=\sum_{x,y}\sqrt{P_{X,Y}(x,y)}\ket{x,y}_{AB}
\otimes e^{i\theta(k,x,y)}\ket{e^{x,y}_k}_{A'}\ket{\mathbf{0}}_{B'}$ be 
an embedding of $P_{X,Y}$ locally equivalent to $\ket{\psi_h}$ but with a trivial register $B'$.

Let us put things together:
\begin{align}
S(X;BB')_\psi &\geq S(X;BB')_{\sum_h \lambda_h \proj{\psi_{h}}} \label{localtcp} \\
                     &= \sum_h \lambda_h  S(X;BB')_{\psi_{h}} \nonumber \\
                      & \geq \min_h{S(X;BB')_{\psi_h}} \nonumber \\
                      &= S(X;BB')_{\widehat{\psi}_{h^*}} \label{localequiv} \enspace , 
\end{align} 
where  (\ref{localtcp}) 
follows from the fact that the local measurement $\mathcal{M}$ does not increase mutual 
information~\cite[Theorem 11.15(3)]{NC00}, and (\ref{localequiv}) follows since 
$\ket{\widehat{\psi}_h}$ is locally equivalent
to $\ket{\psi_h}$ for all $h$. Since $\ket{\widehat{\psi}_{h^*}}$ 
is an embedding of $P_{X,Y}$ with register $B'$ 
being trivial, we can use the reasoning from the beginning of the proof that  
$\ket{\widehat{\psi}_{h^*}}$ is locally equivalent to a state $\ket{\psi^*}_{AB} \otimes \ket{\sigma}_{B'}$ with $\ket{\psi^*}\in \emb{P_{X,Y}}$.
By  Lemma~\ref{symmetry}, the same proof applies to $S(Y;AA')_\psi$. 
\qed
\end{proof}

So far, we have defined the leakage of an embedding of a primitive. We now
define the leakage of a primitive the  natural way:
\begin{definition}
\label{infimum}
We define the \emph{leakage of a primitive $P_{X,Y}$} as the minimal
leakage among all protocols \correctly\ implementing
$P_{X,Y}$. Formally, 
$$\Delta_{P_{X,Y}} \assign \min_{\ket{\psi}}\Delta_{\psi}(P_{X,Y}) \enspace,$$
where the minimization is over all embeddings $\ket{\psi}$ of $P_{X,Y}$.
\end{definition}
Notice that the minimum in the previous definition is well-defined,
because by Lemma~\ref{super_leak}, it is sufficient to minimize over regular embeddings
$\ket{\psi}\in\emb{P_{X,Y}}$. Furthermore, the function
$\Delta_{\psi}(P_{X,Y})$ is continuous on the compact (i.e.~closed and
bounded) set $[0,2\pi]^{|\mathcal{X}\times\mathcal{Y}|}$ of complex
phases corresponding to elements $\ket{x,y}_{AB}$ in the formula for
$\ket{\psi}_{AB}\in \emb{P_{X,Y}}$ and therefore it achieves its
minimum.

The following theorem shows that the leakage of any embedding of a
primitive $P_{X,Y}$ is lower-bounded by the minimal leakage achievable
for primitive $P_{\dep{X}{Y},\dep{Y}{X}}$ (which due to
Lemma~\ref{super_leak} is achieved by a regular embedding).
\begin{theorem}
\label{specialform2}
For any primitive $P_{X,Y}$, 
$\Delta_{P_{X,Y}} \geq \Delta_{P_{\dep{X}{Y},\dep{Y}{X}}}.$
\end{theorem}
\begin{proof}
In fact, the random variables $\dep{X}{Y}$ and $\dep{Y}{X}$ in the
claim can be replaced by any variables $X'$ and $Y'$ with the property that
$X\leftrightarrow X'\leftrightarrow Y$ and $X\leftrightarrow
Y'\leftrightarrow Y$ are Markov chains, and that $Y'=f_Y(Y)$ and
$X'=f_X(X)$ for some deterministic functions $f_Y$ and $f_X$. For such
random variables we then have $I(X';Y')=I(X;Y)$. Therefore, showing
that for $\ket{\psi}\in\emb{P_{X,Y}}$ with the lowest leakage among all regular embeddings of $P_{X,Y}$ (regularity follows from Lemma~\ref{super_leak}) 
and for some $\ket{\psi^*}\in\emb{P_{X',Y'}}$ , it holds that
$$S(A)_\psi-I(X;Y)=\Delta_\psi(P_{X,Y})\geq
\Delta_{\psi^*}(P_{X',Y'})=S(A)_{\psi^*}-I(X';Y')$$ 
is equivalent to proving $S(A)_\psi\geq S(A)_{\psi*}$. First, we show
that there exists $\ket{\tilde{\psi}}\in\emb{P_{X,Y'}}$ such that
$S(A)_{\psi}\geq S(A)_{\tilde{\psi}}$, i.e. $\Delta_\psi(P_{X,Y})\geq
\Delta_{\tilde{\psi}}(P_{X,Y'})$. The existence of $\ket{\psi^*}$ such
that $\Delta_{\tilde{\psi}}(P_{X,Y'})\geq \Delta_{\psi^*}(P_{X',Y'})$
follows from an analogous argument.

State $\ket{\psi}$ can be written in the form:
$$\ket{\psi}=\sum_{x,y}\sqrt{P_{X,Y}(x,y)}e^{i\theta(x,y)}\ket{x,y}_{AB} \enspace.$$

For any realization $y'$ of $Y'$, let $O_{y'}$ be the set of elements $y$ which are mapped to $y'$ under $f_Y(\cdot)$, i.e. $O_{y'} \assign \{y:\ f_{Y}(y)=y'\}$. Let $g$ denote the bijection mapping tuples $(y',j_y) \in Y' \times O_{y'}$ back to $y$. There exists an isometry $U$ on Bob's side such that 
\begin{align}
(\id_A\otimes U)&\ket{\psi}_{AB}
=\sum_{x,y}\sqrt{P_{X,Y}(x,y)}e^{i\theta(x,y)}\ket{x,f_{Y}(y)j_{y}}_{A
  \tilde{B}
  \tilde{B}^{\perp}} \nonumber\\
&=
\sum_{x,y'}\sqrt{P_{X,Y'}(x,y')}\ket{x,y'}_{A \tilde{B}} \sum_{j \in O_{y'}} \sqrt{P_{Y|Y'=y'}\big(g(y',j)\big)}e^{i\theta(x,g(y',j))}\ket{j}_{\tilde{B}^{\perp}}
\, , \label{eq:wantswap}
\end{align}
where $\hil_{\tilde{B}\tilde{B}^{\perp}}\cong \hil_{B}$.

Our goal for the rest of the proof is to transform the register containing 
$j$ into a form where the order of the summations over $(x,y')$ and $j$ in \eqref{eq:wantswap} can be reversed 
to get a state of the
form
\begin{align} \label{eq:swapped}
\ket{\varphi} = \frac{1}{\sqrt{t}}\sum_{j=1}^t
\ket{\hat{\psi}_j}_{A\tilde{B}}\ket{j}_{B'} \enspace ,
\end{align}
where $t$ is some normalization factor and each $\ket{\hat{\psi}_j}$ is in
$\emb{P_{X,Y'}}$. Due to concavity of the von Neumann entropy, we can then argue that
\begin{equation} \label{eq:switchsum} 
\frac{1}{t} \sum_j S\left(\tr_{\tilde{B}B'} \proj{\hat{\psi_j}} \right) \leq S \left(
\frac{1}{t} \sum_j \tr_{\tilde{B}B'} \proj{\hat{\psi_j}} \right) = S\left( \tr_{\tilde{B}B'} \proj{\varphi} \right) = S(A)_\varphi \enspace.
\end{equation}
Hence, there exists a $j$ such that $\ket{\tilde{\psi_j}} \in \emb{P_{X,Y'}}$ and
$S(A)_{\tilde{\psi_j}} \leq S(A)_\varphi = S(A)_\psi$, proving the claim.
\medskip

Let us fix $\delta > 0$ and we show the existence of an embedding
$\ket{\tilde{\psi_j}} \in \emb{P_{X,Y'}}$ such that
$S(A)_{\tilde{\psi_j}} \leq S(A)_\psi + \delta$. In order to reverse the
order of summation in \eqref{eq:wantswap}, we show the existence of an
isometry $W$ on Bob's system such that
$$\ket{\hat{\varphi}} := (\id_A\otimes W)(\id_A\otimes
U)\ket{\psi}_{AB} = \frac{1}{\sqrt{t}}\sum_{z=1}^t
\ket{\hat{\psi}_z}_{A\tilde{B}}\ket{z}_{B'} \enspace,$$
where each $\ket{\hat{\psi}_z}$ is a quantum embedding of a primitive $P_{\hat{X},\hat{Y}}$ that is $\eps$-close (in statistical distance) to the primitive $P_{X,Y'}$. 

The idea is for a given $y'$ and $j \in O_{y'}$ to ``slice up'' the
term $\ket{j}_{\tilde{B}^{\perp}}$ with weight
$\sqrt{P_{Y|Y'=y'}\big(g(y',j)\big)}$ into a lot of very small pieces
of weight $1/\sqrt{t}$ by letting $W$ map
$\ket{j}_{\tilde{B}^{\perp}}$ into superpositions $\sum_z
\ket{z}_{B'}$, where $t \in \mathbb{N}$ is a large natural number to
be determined later as a function of $\delta$. More formally, let us fix $y'$ and denote the elements of the set $O_{y'}$ as $\{1,2,\ldots,k\}$. As a shorthand, we use $p_j := P_{Y|Y'=y'}(g(y',j))$ and note that $\sum_{j=1}^k p_j = 1$. We define $n_j := \lceil t \cdot p_j \rfloor$ to be the natural number of pieces required to approximate $p_j \approx \frac{n_j}{t}$ for large $t$. Let $t_0:=0$ and $t_j := \sum_{i \leq j} n_i$. Then, we define $W$ to map $\ket{j}_{\tilde{B}^{\perp}}$ to $\frac{1}{\sqrt{n_j}} \sum_{z=t_{j-1}+1}^{t_j} \ket{z}$ and get 
\begin{align*}
\ket{\hat{\varphi}} &= 
(\id_A\otimes W) \sum_{x,y'}\sqrt{P_{X,Y'}(x,y')}\ket{x,y'}_{A \tilde{B}} \sum_{j \in O_{y'}} \sqrt{p_j} e^{i\theta(x,g(y',j))}\ket{j}_{\tilde{B}^{\perp}} \\
&= \sum_{x,y'}\sqrt{P_{X,Y'}(x,y')}\ket{x,y'}_{A \tilde{B}} \sum_{j \in O_{y'}} \sqrt{\frac{p_j}{n_j}} e^{i\theta(x,g(y',j))} \sum_{z=t_{j-1}+1}^{t_j} \ket{z} \enspace, 
\end{align*}
It is not hard to verify~\cite{Sotakova08} that $\frac{p_j}{n_j}$ can be written as $\frac{1}{t} + \frac{\eps(y',z)}{t^2}$ where the error $|\eps(y',z)| \leq c$ is upper bounded by a constant $c$ independent of $t$.
Then, we get
\begin{align*} 
\ket{\hat{\varphi}} = &\sum_{x,y'} \sqrt{P_{X,Y'}(x,y')} \ket{x,y'}_{A\tilde{B}} \sum_{z=1}^t  \sqrt{\frac{1}{t} + \frac{\eps(y',z)}{t^2}} e^{i \theta'(x,y',z)}\ket{z}_{B'} \\
&=\frac{1}{\sqrt{t}}\sum_{z=1}^{t}\left(\sum_{x,y'}e^{i\theta'(x,y',z)} \sqrt{1+ \frac{\eps(y',z)}{t}}  \sqrt{P_{X,Y'}(x,y')}\ket{x,y'}_{A\tilde{B}}\right)\ket{z}_{B'}\\
&=\frac{1}{\sqrt{t}}\sum_{z=1}^{t} \ket{\hat{\psi}_z}_{A\tilde{B}} \ket{z}_{B'} \enspace ,
\end{align*}
where $\theta'(x,y',z) = \theta(x,y)$ for $y$ corresponding to
$(y',z)$. Using the reasoning from \eqref{eq:switchsum}, we derive the
existence of a $z$, such that the state $\ket{\hat{\psi}_z} \in
\emb{P_{\hat{X},\hat{Y}}}$ is a regular embedding of a primitive $P_{\hat{X},\hat{Y}}$ that is $\eps(t)$-close to
$P_{X,Y'}$ and $\eps(t) \rightarrow 0$ when $t \rightarrow
\infty$. Furthermore, we have that $S(A)_{\hat{\psi}_z} \leq
S(A)_{\psi}$. As $\ket{\hat{\psi}_z}$ is a regular embedding, we can
write $\ket{\hat{\psi}_z} = \sum_{\hat{x},\hat{y}} \sqrt{
  P_{\hat{X},\hat{Y}}(\hat{x},\hat{y})} e^{i
  \hat{\theta}(\hat{x},\hat{y})} \ket{x} \ket{y}$ for some phase
function $\hat{\theta}(\hat{x}, \hat{y})$.  We define the
desired state $\ket{\tilde{\psi}} \in \emb{P_{X,Y'}}$  as
$\ket{\tilde{\psi}} := \sum_{x,y} \sqrt{P_{X,Y'}(x,y)} e^{i
  \hat{\theta}(x,y)} \ket{x} \ket{y}$. We can choose $t$ large enough
such that the distance $\left\| \ket{\hat{\psi}_z} - \ket{\tilde{\psi}} \right\|$
is arbitrarily small and hence, by the continuity of the von
Neumann entropy, also their entropies $S(A)$ differ by at most
$\delta$. Hence, $S(A)_{\tilde{\psi}} \leq S(A)_{\hat{\psi}_z} + \delta \leq
S(A)_{\psi} + \delta$, which is what we wanted to show. \qed

\end{proof}


\subsection{Leakage as Measure of Privacy and Hardness of
  Implementation} \label{sec:limitedresources} The main results of
this section are consequences of the Holevo bound
(Theorem~\ref{holevo}).
\begin{theorem} \label{thm:nonleaktrivial} 
If a two-party \correct\ quantum protocol for $P_{X,Y}$ does not
leak  then $P_{X,Y}$
is a trivial primitive.
\end{theorem}
\begin{proof} Theorem~\ref{specialform2} implies that if there is a
  $0$--leaking embedding of $P_{X,Y}$ then there is also a
  $0$--leaking embedding of $P_{\dep{X}{Y},\dep{Y}{X}}$. Let us
  therefore assume that $\ket{\psi}$ is a non-leaking embedding of
  $P_{X,Y}$ such that $X\equiv \dep{X}{Y}$ and $Y\equiv \dep{Y}{X}$. We can write
  $\ket{\psi}$ in the form
  $\ket{\psi}=\sum_{x}\sqrt{P_X(x)}\ket{x}\ket{\varphi_x}$ and get
  $\rho_B = \sum_x P_X(x)\proj{\varphi_x}$. For the leakage of
  $\ket{\psi}$ we have:
  $\Delta_\psi(P_{X,Y})=S(X;B)-I(X;Y)=S(\rho_B)-I(X;Y)=0$. From the
  Holevo bound (Theorem~\ref{holevo}) follows that the states
  $\{\ket{\varphi_x}\}_x$ form an orthonormal basis of their span
  (since $X\equiv \dep{X}{Y}$, they are all different) and that $Y$ captures
  the result of a measurement in this basis, which therefore is the
  computational basis. Since $Y\equiv \dep{Y}{X}$, we get that for each $x$,
  there is a single $y_x\in\mathcal{Y}$ such that
  $\ket{\varphi_x}=\ket{y_x}$.
Primitives $P_{\dep{X}{Y},\dep{Y}{X}}$ and $P_{X,Y}$ are therefore
trivial. \qed
\end{proof}

In other words, the only primitives that two-party quantum protocols
can implement \correctly\ (without the help of a trusted third party)
without leakage are the trivial ones! We note also that \correctness\ 
is not required  for Theorem~\ref{thm:nonleaktrivial} to be true. A slightly more involved
proof can be done solely based on the correct distribution of the output
values. This result can be seen as a quantum extension of the
corresponding characterization for the cryptographic power of
classical protocols in the HBC model. \emph{Whereas classical two-party
protocols cannot achieve anything non-trivial, their quantum
counterparts necessarily leak information when they implement
non-trivial primitives.}

\subsection{Tripartite Embeddings} 
\label{app:tripartite}

In this section, we extend the notion of leakage to protocols involving a trusted third party. A special case of such protocols are the ones where the players are allowed one call to a black box who provides them with classical variables $\tilde{X},\tilde{Y}$ sampled according to distribution $P_{\tilde{X},\tilde{Y}}$. It is natural to ask which primitives $P_{X,Y}$ can be implemented without leakage in this case.

The state produced by purifying Alice's and Bob's actions in such a protocol up to the final measurement yielding $X$ and $Y$ can without loss of generality be viewed as a pure state shared among Alice, Bob and an
environment $\ket{\psi}_{EABA'B'}= \sum_e
\sqrt{P_E(e)}\ket{e}_E\otimes\ket{\psi^e}_{ABA'B'}$ . 
We define tripartite embeddings of a
primitive $P_{X,Y}$ analogously to the case of embeddings:

\begin{definition}
\label{corr_impl}
A state $\ket{\psi}=\sum_e P_E(e)\ket{e}_E\otimes
\ket{\psi^e}_{ABA'B'}$ is a \emph{tripartite embedding} of $P_{X,Y}$,
if measuring registers $A$ and $B$ in the computational basis yields
$X,Y$ with distribution $P_{X,Y}$ and the ensemble
$\rho_{ABA'B'}\assign \tr_E\proj{\psi}$ satisfies
$S(X;YB')=S(XA';Y)=I(X;Y)$ .
\end{definition}

The generalization of the notion of leakage to tripartite embeddings
is straightforward:
\begin{definition}
Let $\ket{\psi}\in \hil_{E}\otimes\hil_{ABA'B'}$ be a tripartite embedding of $P_{X,Y}$. 
We define the \emph{leakage} 
of $\rho_{ABA'B'}\assign \tr_E\proj{\psi}$ viewed as an implementation of $P_{X,Y}$ as
$$\Delta_{\rho_{ABA'B'}}(P_{X,Y}):=\max \left\{ S(X;BB')-I(X;Y) \, , \, S(AA';Y)-I(X;Y) \right\}\, .$$
\end{definition}

The leakage of a tripartite embedding is non-negative, for the same reason 
as in the bipartite case. However, it is not necessarily symmetric, for instance for the state $\ket{\psi}_{EAB} = \sqrt{1/3} \left( \ket{001} + \ket{110} + \ket{111} \right)$ which can be verified numerically.

The following theorem shows that non-leaking embeddings of any given 
primitive have the property that Bob's register $\tilde{B}$ holding his dependent part $\dep{Y}{X}$ has to be classical if Alice honestly measures her register $A$ in the computational basis to obtain $X$. An analogous statement holds with the roles of Alice and Bob exchanged. Intuitively, this means that Bob cannot learn more than $\dep{Y}{X}$ about Alice's outcome $X$ from a non-leaking embedding.
\begin{theorem}
\label{th:nonleakprivate}
Let $\ket{\psi}\in \hil_{E}\otimes\hil_{ABA'B'}$  be a non-leaking tripartite embedding of primitive $P_{X,Y}$, where $\hil_{A}=\hil_{\tilde{A}}\otimes \hil_{\tilde{A}^{\perp}}$
and $\hil_{B}=\hil_{\tilde{B}}\otimes\hil_{\tilde{B}^{\perp}}$. Then, there exist unitary transforms
 $U\in \unit{A}$ 
 and $V\in \unit{B}$ such that the state
 $\ket{\psi_{U,V}}= (U\otimes \I_{A'} \otimes V\otimes \I_{B'})\ket{\psi}$ has the following property.
Let
  $\tilde{\rho}_{X \tilde{B}\tilde{B}^{\perp}B'}$ and 
$\tilde{\rho}_{\tilde{A} Y \tilde{A}^{\perp}A'}$ be the states obtained when register $A$
of $\trace[E A']{\proj{{\psi_{U,V}}}}$ is measured in the computational basis to obtain $X$, and register
$B$ of $\trace[E B']{\proj{{\psi_{U,V}}}}$ 
is measured in the computational basis to obtain $Y$.
It then holds that
\[\tilde{\rho}_{X \tilde{B}\tilde{B}^{\perp}B'} = \sum_{x,\tilde{y}} P_{X,\dep{Y}{X}}(x,\tilde{y}) \proj{x}_{X} \otimes \proj{\tilde{y}}_{\tilde{B}} \otimes \sigma^{\tilde{y}}_{\tilde{B}^{\perp}B'}\]
and 
\[
\tilde{\rho}_{\tilde{A} Y \tilde{A}^{\perp}A'} = \sum_{\tilde{x},y} P_{\dep{X}{Y},Y}(\tilde{x},y) \,
\proj{\tilde{x}}_{\tilde{A}} \otimes \proj{y}_{Y} \otimes  \tau^{\tilde{x}}_{\tilde{A}^{\perp}A'} \enspace ,
\]
for some set of density matrices $\{\sigma^{\tilde{y}}\}_{\tilde{y}}$ in $\hil_{\tilde{B}^{\perp}B'}$ and
$\{\tau^{\tilde{x}}\}_{\tilde{x}}$ in $\hil_{\tilde{A}^{\perp}A'}$. 
\end{theorem}

\begin{proof}
Let $\ket{\psi}_{EABA'B'}$ be a tripartite embedding 
of $P_{X,Y}$:
\[
 \ket{\psi} = 
  \sum_{e,x,y,i,j} \sqrt{P_{E,X,Y,I,J}(e,x,y,i,j)}  e^{i\theta(e,x,y,i,j)}\ket{e,x,y,i,j}_{EABA'B'} \enspace .
\]
Let $U$ and $V$ be unitary transforms acting in $\hil_{A}$ and $\hil_{B}$ 
respectively and extracting each party's dependent part ($\dep{X}{Y}$ and $\dep{Y}{X}$ respectively)
in subregisters $\tilde{A}\subseteq A$ and $\tilde{B}\subseteq B$  respectively:
\[ U\ket{x}_A = \ket{f(x)}_{\tilde{A}}\otimes \ket{\mu_x}_{\tilde{A}^{\perp}} \mbox{ and }
   V\ket{y}_{B} = \ket{g(y)}_{\tilde{B}}\otimes \ket{\nu_y}_{\tilde{B}^{\perp}} \enspace ,
\]
for classical functions $f$ and $g$ providing the dependent parts $\dep{X}{Y}$ and $\dep{Y}{X}$
associated to $X$ and $Y$ respectively. For $U$ and $V$ to be unitary, we have that $\bracket{\mu_x}{\mu_{x'}}=0$
for all $x\neq x'$ such that $f(x)=f(x')$, and that $\bracket{\nu_y}{\nu_{y'}}=0$ for all
$y\neq y'$ such that $g(y)=g(y')$.
We define
\[
\begin{split}
 \ket{\tilde{\psi}} & = (\I_{E}  \otimes   U\otimes   \I_{A'} \otimes V \otimes \I_{B'})\ket{\psi} \\
 & =
\sum_{e,x,y,i,j} \sqrt{P_{E,X,Y,I,J}(e,x,y,i,j)}  
e^{i\theta(e,x,y,i,j)}\ket{e,f(x),g(y)}_{\tilde{A}\tilde{B}} \otimes \\
& \hspace{3.7in}
\ket{\mu_x}_{\tilde{A}^{\perp}}\ket{\nu_y}_{\tilde{B}^{\perp}}\ket{i,j}_{A'B'}\enspace ,
\end{split}
\]
where $\hil_{\tilde{A}\tilde{A}^{\perp}}\cong \hil_{A}$ and
$\hil_{\tilde{B}\tilde{B}^{\perp}}\cong \hil_{B}$. 
Re-writing $\ket{\tilde{\psi}}$ in terms of the different values which $X$ may take results in:
\begin{align*}
 \ket{\tilde{\psi}}  &= \sum_{x} \sqrt{P_X(x)} \ket{x}_A
\otimes 
\sum_{e,i} \sqrt{P_{E,I |X=x}(e,i)}
\ket{e,i}_{EA'} \ket{\varphi_{e,{x},i}}_{BB'} \\
 &=: \sum_{x} \sqrt{P_{X}(x)} \ket{x}_{A}\otimes\ket{\zeta_{x}}_{A'E\tilde{B}\tilde{B}^{\perp}B'} \enspace .
\end{align*}
We can view the information provided to Bob about $X$ as the information
available about $X$ when encoded by  
 $x \mapsto \rho_{x}$ where:
\begin{align*}
\rho_{x}  &= \trace[A'E]{\proj{\zeta_{x}}}=
\sum_{e,i} P_{E,I |X=x}(e,i) \proj{\varphi_{e,i}}_{BB'}\enspace .
\end{align*}

By basic properties of the von Neumann entropy, we have that 
\begin{align*}
S(X;BB')_\psi &= S(X;BB')_{\tilde{\psi}}\\
&=S(X;\tilde{B}\tilde{B}^\perp B')_{\tilde{\psi}}\\
&\geq S(X;\dep{Y}{X} \tilde{B}^\perp B')_{\tilde{\psi}}\\  
&\geq I(X;\dep{Y}{X}) \, .
\end{align*} 
Suppose now that $\ket{{\psi}}$ is non-leaking, that is $S(X;BB')=I(X;Y)=I(X;\dep{Y}{X})$. It follows that for non-leaking embeddings, all terms above are actually equal. 

By the Holevo bound (Theorem~\ref{holevo}), we conclude that 
states in $\{\rho_{x}\}_{x}$ are simultaneously diagonalizable. In other words, 
for all $x$, 
\[ \rho_{x} = \sum_{z} P_{Z|X=x}(z) \proj{\gamma_z}_{BB'} \enspace ,
\] 
where $\{\ket{\gamma_z}\}_z$ form an orthonormal basis for some subspace
of $\hil_{\tilde{B}\tilde{B}^\perp B'}$. Since $S(X;BB')_{\tilde{\psi}}=
S(X;\dep{Y}{X}\tilde{B}^{\perp}B')_{\tilde{\psi}}$, we conclude that such a basis can be chosen to be the computational basis for the register $\tilde{B}$ holding $\dep{Y}{X}$:
\begin{eqnarray*}
\tilde{\rho}_{X \tilde{B}\tilde{B}^{\perp}B'} &=& \sum_{x} P_{X}(x) \proj{x}_{X} \otimes \rho_{x}\\
&=& \sum _{x} P_{X}(x) \proj{x}_{X} \otimes \sum_z
P_{Z|X=x}(z) \, \proj{\gamma_z}_{BB'}\\
&=&\sum_{x} P_{X}(x) \proj{x}_{X} \otimes 
\sum_{\tilde{y}} P_{\dep{Y}{X}|X=x}(\tilde{y}) \,
\proj{\tilde{y}}_{\tilde{B}} \otimes \sigma^{\tilde{y},x}_{\tilde{B}^{\perp}B'} \\
&=& \sum_{x,\tilde{y}} P_{X,\dep{Y}{X}}(x,\tilde{y}) 
\proj{x}_{X} \otimes 
\proj{\tilde{y}}_{\tilde{B}}\otimes \sigma^{\tilde{y},x}_{\tilde{B}^{\perp}B'} \enspace .
\end{eqnarray*}
We now observe that $\ket{\psi}$ being non-leaking implies that $\sigma^{\tilde{y},x}_{\tilde{B}^{\perp}B'}$ 
cannot depend on $x$. Otherwise, suppose that for some $\tilde{y}$ there exist
$x\neq x'$ such that $\sigma^{\tilde{y},x}_{\tilde{B}^{\perp}B'}\neq 
\sigma^{\tilde{y},x'}_{\tilde{B}^{\perp}B'}$ with $P_{X,\dep{Y}{X}}(x,\tilde{y})>0$
and $P_{X,\dep{Y}{X}}(x',\tilde{y})>0$. After having measured $\tilde{y}$, Bob can apply an optimal 
measurement for distinguishing between $\sigma^{\tilde{y},x}_{\tilde{B}^{\perp}B'}$ and 
$\sigma^{\tilde{y},x'}_{\tilde{B}^{\perp}B'}$ with some strictly positive bias 
allowing him to get more information than $I(X;\dep{Y}{X})$ thereby implying that $\ket{\psi}$
is leaking. It follows that
\begin{equation*}
\tilde{\rho}_{X \tilde{B}\tilde{B}^{\perp} B'} 
=  \sum_{x,\tilde{y}} P_{X,\dep{Y}{X}}(x,\tilde{y}) 
\proj{x}_{X} \otimes \proj{\tilde{y}}_{\tilde{B}}\otimes \sigma^{\tilde{y}}_{\tilde{B}^{\perp}B'} \enspace .
\end{equation*}
The same argument symmetrically applied to $\rho_{\tilde{A}Y\tilde{A}^{\perp}A'}$ leads to 
\begin{equation*}
 \tilde{\rho}_{\tilde{A} Y \tilde{A}^{\perp} A'} = \sum_{\tilde{x},y} P_{\dep{X}{Y},Y}(\tilde{x},y)\, 
\proj{\tilde{x}}_{\tilde{A}} \otimes  \proj{y}_{Y} \otimes \tau^{x}_{\tilde{A}^{\perp}A'}  \enspace .
\end{equation*}
\qed
\end{proof}

For the remainder of this section, we focus on primitives $P_{{X},{Y}}$
where each variable is equivalent to its dependent part: $X\equiv \dep{X}{Y}$ and $Y\equiv \dep{Y}{X}$.
For non-leaking tripartite embeddings of these primitives, we establish lower bounds on the conditional von Neumann entropy
of the environment given each party's quantum states.

In order to define what we mean by the entropy of the environment, we decompose any tripartite
embedding $\ket{\psi}$ of $P_{X,Y}$ in its Schmidt form with respect to
the environment:
\[ \ket{\psi}_{EABA'B'} = \sum_{w} \sqrt{\lambda_w}\ket{e_w}_E \otimes \ket{\psi_w}_{ABA'B'}\enspace ,
\]
where $\bracket{e_w}{e_{w'}}=\bracket{\psi_w}{\psi_{w'}}=\delta_{w,w'}$.
Now, imagine the environment measures register $E$
 in the Schmidt basis $\{\proj{e_w}\}_{w}$ to get classical random outcome $W$
 such that
 $\Pr[W=w]=\lambda_w$. 
Corollary~\ref{cor:seab} below shows that non-leaking tripartite embeddings
of $P_{X,Y}$ must satisfy $S(W|\,A\,A')\geq H(Y|\, X)$ and 
$S(W|\,B\,B')\geq H(X|\,Y)$. If the Schmidt decomposition is not
unique, the result holds for a measurement in any Schmidt basis.
Measurements in the Schmidt basis minimize the
entropy of the outcome among any complete Von Neumann measurement 
applied to the state of the environment. 
Intuitively, $S(W|\,A\,A')$ measures the amount of
\emph{shared entanglement} between Alice and the environment
(similarly,  $S(W|\,B\,B')$ is a measure for the shared entanglement
between Bob and the environment). The more non-trivial a primitive gets, 
the more the environment has to be entangled with the players in order to preserve privacy. 

\begin{corollary} \label{cor:seab}
Let $\ket{\psi}_{EABA'B'}$ be any non-leaking tripartite embedding of 
$P_{X,Y}$ where $X\equiv \dep{X}{Y}$ and $Y\equiv \dep{Y}{X}$. 
Let $W$ be the random variable 
for the outcome of measuring the environment register $E$ in a Schmidt basis.
Then,
\begin{align*}
S(W|\,A\, A') &= S(W|\,X\, A') \geq
H(Y|\,X)=H(\dep{Y}{X}|\,\dep{X}{Y}) \mbox{ and }\\
 S(W|\,B\, B') &=S(W|\,Y\, B')\geq H(X|\,Y)=H(\dep{X}{Y}|\,\dep{Y}{X}) \enspace .
\end{align*}
\end{corollary}
\begin{proof}
Let us write the non-leaking tripartite embedding as a function of Alice's 
output $X=x$ as follows
\begin{align}
 \ket{\psi}_{ABEA'B'} &=\sum_x \sqrt{P_{X}(x)}\ket{x}_A\otimes \sum_y 
\sqrt{P_{Y|{X}=x}(y)}\  \ket{y}_{B}\otimes \sum_{a} \kappa_a^{x,y} \,\ket{a}_{A'} \otimes \ket{\mu^{x,y,a}}_{EB'}
\enspace , \label{trie2} 
\end{align}
where we assume without loss of generality that all $\kappa_a^{x,y}$
are real (and possible complex phases are put into
$\ket{\mu^{x,y,a}}$). We claim that $\kappa^{x,y}_a=\kappa^x_a$, i.e.~the coefficients
do not depend on $y$. To see this, let $p^{x}_a=\Pr{({\sf A'}=a|X=x)}$ where ${\sf A'}$
denotes the classical outcome of measuring register $A'$ in the
computational basis. Suppose for a contradiction
that there exists $y$
such that $|\kappa^{x,y}_a|^2 \neq p^{x}_a$:
\begin{align*} 
\Pr{\left(Y=y|X=x,{\sf{A'}}=a\right)} &= \frac{\Pr{\left(Y=y|X=x\right)}\Pr{\left({\sf{A'}}=a|X=x, Y=y\right)}}
{\Pr{\left({\sf{A'}}=a|X=x \right)}} \\
&= \frac{\Pr{\left(Y=y|X=x\right)}|\kappa^{x,y}_a|^2}{p^x_a}\\
&\neq \Pr{\left(Y=y|X=x\right)} \enspace ,
\end{align*}
which would contradict the fact that $\ket{\psi}$ is
non-leaking. Hence, we can write $\ket{\psi}$ as
\begin{equation*}
 \ket{\psi}_{ABEA'B'} =\sum_x \sqrt{P_{X}(x)}\ket{x}_A \sum_a
 \kappa_a^x \ket{a}_{A'} \ket{\eta^{x,a}}_{BEB'} 
\enspace , 
\end{equation*}
where 
\begin{align*}
 \ket{\eta^{x,a}}_{BEB'} &=\sum_y \sqrt{P_{Y|X=x}(y)} \ket{y}_B
 \ket{\mu^{x,y,a}} \\
&= \sum_y \sqrt{P_{Y|X=x}(y)} \ket{y}_B \sum_i \sqrt{\lambda^{x,y,a}_i} \ket{e^{x,y,a}_i}_E\otimes \ket{b^{x,y,a}_i}_{B'}
\enspace , 
\end{align*}
where in the last step, we wrote the bipartite states  $\ket{\mu^{x,y,a}}_{EB'}$ in the
Schmidt form.

Theorem~\ref{th:nonleakprivate} establishes that if $\ket{\psi}$ is non-leaking then 
\begin{equation}\label{req}
 \trace[E]{\proj{\eta^{x,a}}_{BEB'}} = \sum_y P_{Y|X=x}(y) \proj{y}_B \otimes \sigma^y_{B'} \enspace .
\end{equation}
We claim that \eqref{req} implies that the subspaces $S_a^{x,y} =
\mathrm{span}_i \{ \ket{e^{x,y,a}_i}_E \}$ must be perpendicular for
different values of $y$.
Let $q^x_y := \sqrt{P_{Y|X=x}(y)}$ be used to shorten the notation. We have 
\begin{align}
\trace[E]{\proj{\eta^{x,a}}_{BEB'}} &= \sum_{y,y'} q^x_y q^x_{y'}
                                           \trace[E]{\ketbra{y}{y'}_B \otimes \ketbra{\mu^{x,y,a}}{\mu^{x,y',a}}_{EB'}}\nonumber \\
&= \sum_{y,y'} q^x_yq^x_{y'}
\trace[E]{\ketbra{y}{y'}_B \otimes \sum_{i,j} \ketbra{e^{x,y,a}_i}{e^{x,y',a}_j}_E \otimes \ketbra{b^{x,y,a}_i}{b^{x,y',a}_j}_{B'}} \nonumber \\
&= \sum_{y,y'} q^x_yq^x_{y'}\ketbra{y}{y'}_B \otimes 
     \trace[E]{ \sum_{i,j} \ketbra{e^{x,y,a}_i}{e^{x,y',a}_j}_E \otimes \ketbra{b^{x,y,a}_i}{b^{x,y',a}_j}_{B'}}\nonumber \\
&= \sum_{y,y'} q^x_yq^x_{y'}\ketbra{y}{y'}_B \otimes \nonumber \\
&\phantom{=} \hspace{0.2in} \sum_h \bra{e^{x,y,a}_h} \left( \sum_{i,j} \ketbra{e^{x,y,a}_i}{e^{x,y',a}_j}_E \otimes \ketbra{b^{x,y,a}_i}{b^{x,y',a}_j}_{B'} \right) \ket{e^{x,y,a}_h}\nonumber \\
&= \sum_{y,y'} q^x_yq^x_{y'}\ketbra{y}{y'}_B \otimes 
\sum_{i,j} \bracket{e^{x,y',a}_j}{e^{x,y,a}_i} \otimes \ketbra{b^{x,y,a}_i}{b^{x,y',a}_j}_{B'} \enspace . \label{conc}
\end{align}
Clearly, if $S_a^{x,y}\perp S_a^{x,y'}$ is not
satisfied then there exists $i\neq j,y\neq y'$ such that $\bracket{e^{x,y',a}_j}{e^{x,y,a}_i}\neq 0$ and
register $B$ is not diagonal in the computational basis according (\ref{conc}). 

It follows that for $X=x$, any $\mathsf{A'}=a$, and when $Y=y$ is measured by Bob,  the environment $E$
ends up in subspace $S_a^{x,y}$ of $E$ which corresponds to $Y=y$
unambiguously. As $W$ is the outcome of measuring
$E$ in the Schmidt basis, knowledge of $W$, $X$ and $\mathsf{A'}$ determines
$Y$. Formally, we have $0 \leq S(Y|WXA') \leq S(Y|WX\mathsf{A'})=0$ and it follows that
\begin{align} \begin{split}
S(W|XA') &= S(W|XA') + S(Y|WXA')\\
&= S(WY|XA') \\
&\geq S(Y|XA') \\
&= H(Y|X) \enspace , \end{split} \label{swxy1}
\end{align}
where the inequality holds due to the classicality of $W$ and the last step is due to \correctness.

The same argument with the roles of Alice and Bob reversed results in:
\begin{equation}\label{seyx2}
S(W|\,{Y}\, B') \geq H({X}|\,{Y}) \enspace .
\end{equation}
Equations (\ref{swxy1}) and (\ref{seyx2}) establish the inequalities
of the statement.

To prove the statement's equalities, 
consider  Theorem~\ref{th:nonleakprivate} 
when Bob measures $Y$:
\[ \tilde{\rho}_{A'AY} = \sum_{x,y} P_{X,Y}(x,y)\,
  \tau^x_{A'}\otimes \proj{x,y}_{AY}\enspace ,
\]
which obviously means that
\[ \tilde{\rho}_{A'A} = \sum_x P_{X}(x) \, \tau^x_{A'} \otimes \proj{x}_A\enspace .
\]
Alice's register $A$ is therefore diagonal in the computational basis.
It follows that
\[ S(W|\,A\,A') = S(W|\, X\, A')\enspace .
\]
A symmetric argument from (\ref{seyx2}) shows that
\[ S(W|\, B\, B') = S(W|\, Y\, B')\enspace .
\]
\qed
\end{proof}

Suppose Alice and Bob have access to an ideal functionality for $P_{X,Y}$
as a cryptographic resource. What primitives can Alice and Bob implement
without leakage given access to this resource? Is it possible for them to
``promote'' the ideal functionality for $P_{X,Y}$ to a stronger cryptographic primitive?
Before answering this question in the negative,
let us define what we exactly mean by an \emph{ideal functionality} for primitive $P_{X,Y}$:
\begin{definition}
An \emph{ideal functionality \idf{P_{X,Y}} for primitive $P_{X,Y}$}
is a box that provides Alice and 
Bob with $X$ and $Y$ respectively and nothing more. In particular, the ideal functionality
never provides extra working registers (otherwise, extra registers could without violating \correctness\
provide additional cryptographic resources to Alice and Bob). More formally,
\[ \idf{P_{X,Y}} = \sum_{x,y} P_{X,Y}(x,y) 
\proj{x}_A \otimes \proj{y}_B\enspace .
\] 
\end{definition}

The next theorem shows that one call to an ideal functionality is never sufficient for a non-leaking implementation of a stronger primitive. In other words, quantum communication and computation never allow to amplify  an ideal classical two-party cryptographic primitive into a stronger one without leakage.

\begin{theorem}
\label{no-amplif}
Let $P_{{X},{Y}}$ and $P_{{X'},{Y'}}$ 
be two primitives, where $X\equiv \dep{X}{Y}$, $Y\equiv \dep{Y}{X}$, $X'\equiv \dep{X'}{Y'}$, and
$Y'\equiv \dep{Y'}{X'}$. 
Suppose that  $H({X'}|\,{Y'})>H({X}|\, {Y})$ 
or $H({Y'}|\, {X'})>H({Y}|{X})$. 
Then, any implementation of $P_{{X'},{Y'}}$ using just one call 
to the ideal functionality $\idf{P_{{X},{Y}}}$ leaks information. 
\end{theorem}
\begin{proof}
We may view the ideal functionality $\idf{P_{X,Y}}$ as a box that conceals its environment
to Alice and Bob. For instance, the state
\[ \ket{\psi}_{E\hat{A}\hat{B}} = \sum_{x,y} \sqrt{P_{\dep{X}{Y},\dep{Y}{X}}(x,y)} \ket{x,y}_E\otimes\ket{x,y}_{\hat{A}\hat{B}}\enspace .
\]
is a non-leaking embedding of $P_{\dep{X}{Y},\dep{Y}{X}}$ 
with $S(W |\,\hat{A})=H(Y|\,{X})$ and $S(W|\,\hat{B})=H(X|\,Y)$ where
$W$ is defined as above as the classical outcome when measuring the environment
in the Schmidt basis.

Consider any \correct\ protocol implementing $P_{X',Y'}$ where Alice
and Bob purify their actions but are otherwise honest. An execution
of the protocol will produce
a non-leaking tripartite embedding of $P_{X',Y'}$.
Just before the call to $\idf{P_{X,Y}}$, Alice's internal register
$A_0$ and Bob's internal register $B_0$ are such that 
\[ S(W|\,A_0)=S(W|\,B_0)=0 \enspace,
\]
since the environment is in a pure state.
Just after the call to \idf{P_{X,Y}},
Alice's register $A_1$ and Bob's register $B_1$ satisfy:
\[ S(W|\,A_1)=H(Y|\, X)\mbox{ and } S(W|\, B_1)=H(X|\, Y)\enspace ,
\]
since $\hat{A}\subseteq A_1$ and $\hat{B}\subseteq B_1$. Notice
also that the state provided to Alice and Bob by  $\idf{P_{X,Y}}$
is diagonal in the computational basis: the information is classical.
It follows that Alice and Bob can copy this information and keep it 
with them during the execution of the protocol while remaining able
to run the protocol in a honest-but-curious fashion. The Schmidt basis for the environment
remains the same after the call to  \idf{P_{X,Y}}. It follows that 
at any point $t$ in the protocol evolution, Alice's and Bob's internal
quantum registers $A_t$ and $B_t$ respectively are such that: 
\begin{equation}\label{mon}
 S(W|\,A_t)\leq H(Y|\, X)\mbox{ and } S(W|\, B_t)\leq H(X|\, Y)\enspace .
\end{equation}
That is, $S(W|A_t)$ and $S(W|B_t)$ are non-increasing monotones for
honest-but-curious quantum players in secure two-party
computation similar to $H(Y|\, X)$ and $H(X|\, Y)$ in the classical case~\cite{WW04}. 

At the very last step $t_{\tt max}$ of the protocol, $A_{t_{\tt max}} :=A\otimes A'$
and $B_{t_{\tt max}}:=B\otimes B'$.  Therefore,
\[S(W|\,A\,A')\leq H(Y|\, X)\mbox{ and } S(W|\, B\,B')\leq H(X|\, Y)\enspace .
\]
Since $H(Y|\, X)<H(Y'|\, X')$ or $H(X|\, Y)<H(X'|\, Y')$, we conclude that
either $S(W|\, A\,A') < H(Y'|\, X')$ or $S(W|\, B\,B') < H(X'|\, Y')$. It follows by Corollary~\ref{cor:seab}
that the implementation of $P_{X',Y'}$ must leak.
\qed
\end{proof}

As in the classical case~\cite{WW04}, it is straightforward to use Theorem~\ref{no-amplif} in order to determine 
a lower bound on the number of calls to a weaker primitive required to implement
a stronger one without leakage: $P_{X',Y'}$ can be implemented without leakage by $n$ calls
to $P_{X,Y}$ only if $H(X'|\, Y')\leq n H(X|\, Y)$ and $H(Y'|\, X') \leq n H(Y|\, X)$.

\subsection{Reducibility of Primitives and Their Leakage}
This section is concerned with the following question: Given two
primitives $P_{X,Y}$ and $P_{\tilde{X},\tilde{Y}}$ such that $P_{X,Y}$ is reducible
to $P_{\tilde{X},\tilde{Y}}$, what is the relationship between the leakage of
$P_{X,Y}$ and the leakage of $P_{\tilde{X},\tilde{Y}}$?  We use the notion of
reducibility in the following sense: We say that a primitive $P_{X,Y}$
is \emph{reducible in the HBC model} to a primitive $P_{\tilde{X},\tilde{Y}}$ if
$P_{X,Y}$ can be securely implemented in the HBC model from (one call
to) a secure implementation of $P_{\tilde{X},\tilde{Y}}$. The above question can
also be generalized to the case where $P_{X,Y}$ can be computed from
$P_{\tilde{X},\tilde{Y}}$ only with certain probability.  Notice that the answer,
even if we assume perfect reducibility, is not captured in our
previous result from Lemma~\ref{super_leak}, since an embedding of
$P_{\tilde{X},\tilde{Y}}$ is not necessarily an embedding of $P_{X,Y}$ (it might
violate the \correctness\ condition). However, under certain
circumstances, we can show that $\Delta_{P_{\tilde{X},\tilde{Y}}}\geq
\Delta_{P_{X,Y}}.$
\begin{theorem}
\label{reduction}
Assume that primitives $P_{X,Y}$ and $P_{\tilde{X},\tilde{Y}}=P_{\tilde{X}_0\tilde{X}_1,\tilde{Y}_0\tilde{Y}_1}$ satisfy the condition:
$$
\sum_{x,y:P_{\tilde{X}_0,\tilde{Y}_0|\tilde{X}_1=x,\tilde{Y}_1=y}\simeq P_{X,Y}} P_{\tilde{X}_1,\tilde{Y}_1}(x,y)\geq 1-\delta,
$$
where the relation $\simeq$ means that the two distributions are equal
up to relabeling of the alphabet. Then, $\Delta_{P_{\tilde{X},\tilde{Y}}}\geq
(1-\delta)\Delta_{P_{X,Y}}.$
\end{theorem}
\begin{proof}
State $\ket{\psi}_{A_0A_1B_0B_1}\in \emb{P_{\tilde{X},\tilde{Y}}}$ can be written in the form:
$$\ket{\psi}=\sum_{x\in\mathcal{X}'_1} \sqrt{P_{\tilde{X}_1}(x)}\ket{x}_{A_1}\ket{\psi^x}_{A_0B}\, ,$$ 
where each
$\ket{\psi^x}$ 
is a regular embedding of $P_{\tilde{X}_0\tilde{Y}_0\tilde{Y}_1|\tilde{X}_1=x}$.  
Due to the Holevo bound (Theorem~\ref{holevo}), we have
$$S(\tilde{Y}|A)_\psi = S(\tilde{Y}|A_0 A_1)_\psi \leq
S(\tilde{Y}|A_0,\tilde{X}_1)_\psi=\sum_x
P_{\tilde{X}_1}(x)S(\tilde{Y}|A_0,\tilde{X}_1=x)_{\psi^x} \enspace ,$$
and we obtain for the leakage of $\ket{\psi}$ that
\begin{eqnarray*}
\Delta_\psi(P_{\tilde{X},\tilde{Y}})&=&H(\tilde{Y}|\tilde{X})-S(\tilde{Y}|A)_\psi\\
&\geq& H(\tilde{Y}|\tilde{X})-\sum_x P_{\tilde{X}_1}(x)S(\tilde{Y}|A_0,\tilde{X}_1=x)_{\psi^x}\\
&=& \sum_x P_{\tilde{X}_1}(x)(H(\tilde{Y}|\tilde{X}_0,\tilde{X}_1=x)-S(\tilde{Y}|A_0,\tilde{X}_1=x)_{\psi^x})\\
&=& \sum_x P_{\tilde{X}_1}(x)\Delta_{\psi^x}(P_{\tilde{X}_0,\tilde{Y}_0\tilde{Y}_1|\tilde{X}_1=x}) \, .
\end{eqnarray*}

By applying the same argument to each $\ket{\psi^x}$, we obtain that 
\begin{equation}
\label{reduc}
\Delta_\psi(P_{\tilde{X},\tilde{Y}})\geq
\sum_{xy}P_{\tilde{X}_1,\tilde{Y}_1}(x,y)\Delta_{\psi^{x,y}}(P_{\tilde{X}_0,\tilde{Y}_0|\tilde{X}_1=x,\tilde{Y}_1=y})
\, ,
\end{equation}
where each $\ket{\psi^{x,y}}$ is a regular embedding of $P_{\tilde{X}_0,\tilde{Y}_0|\tilde{X}_1=x,\tilde{Y}_1=y}$. 
For each $(x,y)$ such that 
$P_{\tilde{X}_0,\tilde{Y}_0|\tilde{X}_1=x,\tilde{Y}_1=y}\simeq P_{X,Y}$ is satisfied, we get that 
$$\Delta_{\psi^{x,y}}(P_{\tilde{X}_0,\tilde{Y}_0|\tilde{X}_1=x,\tilde{Y}_1=y})\geq \Delta_{P_{X,Y}}\, .$$ 
Since $\sum_{x,y:P_{\tilde{X}_0,\tilde{Y}_0|\tilde{X}_1=x,\tilde{Y}_1=y}\simeq P_{X,Y}} P_{\tilde{X}_1,\tilde{Y}_1}(x,y)\geq 1-\delta,$ 
we get from~(\ref{reduc}) that 
$$\Delta_\psi(P_{\tilde{X},\tilde{Y}})\geq (1-\delta)P_{X,Y}\, .$$ 
\qed
\end{proof}

Theorem~\ref{reduction} will allow to derive a lower bound on the leakage of
1-out-of-2 Oblivious Transfer of $r$-bit strings in
Section~\ref{sec:primleakage}.

\section{The Leakage of Universal Cryptographic
  Primitives}\label{sec:primleakage} 

 In this section, we exhibit lower
bounds on the leakage of the following universal two-party primitives.

\begin{description}
\item[String Rabin OT ($\srot$):]~\cite{Rabin81} Alice sends a random
  string of $r$ bits to Bob who receives it with probability $1/2$,
  otherwise he receives a special symbol $\bot$. Alice does not learn
  any information about whether Bob has received the string she sent. 

For $x\in\{0,1\}^r$ and $y\in\{0,1\}^r\cup \{\bot\}$:
\[\psrot(x,y)=\left\{\begin{array}{ll}{2^{-r-1}} &\mbox{if $x=y$ or $y=\bot$,}\\
                                    0                 &\mbox{otherwise},
\end{array}\right.
\]
is the joint probability distribution associated to an execution of
Rabin OT of a random binary string of length $r$.

\item[One-out-of-two String OT ($\sotr$):]~\cite{Wiesner83,EGL82}
Alice sends two random $r$-bit strings to Bob who decides which of them he receives. Bob does not learn 
any information about the other one of Alice's strings and Alice does not learn which of the strings has been received by Bob. We simply write \ot\ for the case of 1-out-of-2 oblivious transfer of bits ($r=1$).

For $x_0,x_1,y\in\{0,1\}^r$ and $c\in\{0,1\}$:
\[\psot((x_0,x_1),(c,y))=\left\{\begin{array}{ll}{2^{-2r-1}} &\mbox{if $y=x_c$,}\\
                                                            0    &\mbox{otherwise},
\end{array}\right.
\]
is the joint probability distribution associated to an execution of
one-out-of-two $r$-bit string OT upon random inputs.

\item[Additive sharing of AND ($\sand$):]~\cite{PR94}
Alice and Bob choose their respective input bits $x$ and $y$, and receive the output bits $a$ resp. $b$ such that $a\oplus b=x\wedge y$ and $\Pr[a=0]=1/2$. They do not get any other information. 

For $x,y,a,b\in\{0,1\}$:
\[\pnl((x,a),(y,b))=\left\{\begin{array}{ll}\frac{1}{8} & \mbox{if $xy=a\oplus b$,}\\
                                                 0       & \mbox{otherwise}, 
                          \end{array}\right.
\]
is the joint probability distribution associated to the generation of an additive
sharing for the {\sc and} of two random bits.

\item[Noisy one-out-of-two OT ($\otn$):] Alice sends two bits to Bob
 who decides which of them he wants to receive. The selected bit is
 transmitted to him over a noisy channel with noise rate $p$. Bob
 does not learn any information about the other one of Alice's bits
 and Alice does not learn any information about Bob's selection bit.

For $x_0,x_1,y,c\in\{0,1\}$ and $p\in\,(0,1/2)$:
\[\potn((x_0,x_1),(c,y))=\left\{\begin{array}{ll}\frac{1-p}{8} &\mbox{if $y=x_c$,}\\
                                                \frac{p}{8} &\mbox{otherwise},
\end{array}\right.
\]
is the joint probability distribution associated to an execution
of one-out-of-two OT where the selected bit is received  through a binary symmetric
channel with error rate $p$.

\end{description}

Table~\ref{tab:lowerbounds} summarizes the lower bounds on the leakage
of these primitives (the derivations can be found in
Appendix~\ref{app:universalleakage}). We note that Wolf and
Wullschleger~\cite{WW05b} have shown that a randomized \ot\ can be
transformed by local operations into an additive sharing of an AND
(here called \sand).  Therefore, our results for \ot\ below also apply
to \sand.

\begin{table}[h] \center \renewcommand{\arraystretch}{1.5}
\begin{tabular}{@{}l|c|@{\ \ }l@{}} \hline
{\bf primitive} & {\bf leaking at least} & {\bf comments} \\ \hline
$\srotp{1}$ & $(h(\frac{1}{4})-\frac{1}{2}) \approx 0.311$ & same
leakage for all regular embeddings \\ \hline 
$\srot$ & $(1-O(r2^{-r}))$ & same leakage for all regular embeddings\\ \hline
$\ot,\sand$ \hspace{2mm} & $\frac12$ & minimized by canonical embedding\\ \hline 
$\sotr$ & $(1-O(r2^{-r}))$ & (suboptimal) lower bound \\ \hline 
$\otn$ & $\frac{\left(1/2-p-\sqrt{p(1-p)}\right)^2}{8\ln 2}$ \hspace{2mm} &
if $p<\sin^2(\pi/8)\approx 0.15$, (suboptimal) lower bound \\ \hline
\end{tabular}
\medskip
\caption{Lower bounds on the leakage for universal two-party
  primitives} \label{tab:lowerbounds}
\vspace{-0.6cm}
\end{table}

\sotr\ and \otn\ are primitives where the direct evaluation of the
leakage for a general embedding $\ket{\psi_\theta}$ is hard, because
the number of possible phases increases exponentially in the number of
qubits. Instead of computing $S(A)$ directly, we derive (suboptimal)
lower bounds on the leakage.

For the primitive \potn, our lower bound on the leakage only holds for
$p<\sin^2(\pi/8)\approx 0.15$. Notice that in reality, the leakage is
strictly positive for any embedding of \potn\ with $p<1/4$, since for
$p<1/4$, \potn\ is a non-trivial primitive. On the other hand, $\potnq$ is a trivial primitive implemented securely by the following
protocol in the classical HBC model:
\begin{enumerate}
\item Alice chooses randomly between her input bits $x_0$ and $x_1$ and sends the chosen value $x_a$ to Bob. 
\item Bob chooses his selection bit $c$ uniformly at random and sets $y\assign x_a$.
\end{enumerate}
Equality $x_c=y$ is satisfied if either $a=c$, which happens with
probability $1/2$, or if $a\neq c$ and $x_a=x_{1-a}$, which happens
with probability $1/4$. Since the two events are disjoint, it follows
that $x_c=y$ with probability $3/4$ and that the protocol implements
\potnq. The implementation is clearly secure against
honest-but-curious Alice, since she does not receive any message from
Bob. It is also secure against Bob, since he receives only one bit
from Alice. By letting Alice randomize the value of the bit she is
sending, the players can implement \potn\ securely for any value
$1/4<p\leq 1/2$.

\section{Conclusion and Open Problems}\label{conclusion}
We have provided a quantitative extension of qualitative impossibility
results for two-party quantum cryptography. All non-trivial classical primitives
leak information when implemented by quantum protocols.  Notice that
demanding a protocol to be non-leaking does in general not imply the
privacy of the players' outputs.  For instance, consider a protocol
implementing \ot\ but allowing a curious receiver with probability
$\frac12$ to learn both bits simultaneously or with probability
$\frac12$ to learn nothing about them.  Such a protocol for \ot\ would
be non-leaking but nevertheless insecure.  Consequently,
Theorem~\ref{thm:nonleaktrivial} not only tells us that any quantum
protocol implementing a non-trivial primitive must be insecure, but
also that a privacy breach will reveal itself as
leakage. Our framework allows to quantify the leakage of any two-party quantum
protocol \correctly\ implementing a primitive. Our impossibility results
are different than common ones since they only rely on
the \correctness\ of the protocol, not on the perfect privacy
of a protocol against one party . Moreover, the generic attack that allows to 
show leakage is simply implemented by purifying the parties' actions.  
Furthermore, we present
lower bounds on the leakage of some \correct\ universal two-party primitives.

A natural open question is to find a way to identify good embeddings
for a given primitive. 
Based on the examples of \srot\ and \ot, it is tempting to conjecture
the following.
\begin{conjecture}
The leakage of any primitive $P_{X,Y}$ is minimized by its canonical embedding.
\end{conjecture}
The conjecture agrees with the geometric intuition that the minimal pairwise
distinguishability of quantum states in a mixture minimizes the von
Neumann entropy of the mixture. However, Jozsa and Schlienz have shown
that this intuition is sometimes incorrect~\cite{JS00}. In a quantum
system of dimension at least three, we can have the following
situation: For two sets of pure states $\{\ket{u_i}\}_{i=1}^n$ and
$\{\ket{v_i}\}_{i=1}^n$ satisfying $|\bracket{u_i}{u_j}|\leq
|\bracket{v_i}{v_j}|$ for all $i,j$, there exist probabilities $p_i$
such that for $\rho_u \assign \sum_{i=1}^np_i\ketbra{u_i}{u_i}$,
$\rho_v \assign \sum_{i=1}^np_i\ketbra{v_i}{v_i}$, it holds that
$S(\rho_u)<S(\rho_v)$.  As we can see, although each pair $\ket{u_i}$,
$\ket{u_j}$ is more distinguishable than the corresponding pair
$\ket{v_i}$, $\ket{v_j}$, the overall $\rho_u$ provides us with less
uncertainty than $\rho_v$. It follows that although for the canonical
embedding $\ket{\psi_{\vec{0}}}=\sum_y\ket{\varphi_y}\ket{y}$ of
$P_{X,Y}$ the mutual overlaps $|\bracket{\varphi_y}{\varphi_{y'}}|$
are clearly maximized, it does not necessarily imply that $S(A)$ in
this case is minimal over $\emb{P_{X,Y}}$. It is an interesting open
question to find a primitive whose canonical embedding does not
minimize the leakage or to prove that no such primitive exists. In
particular, how far can the leakage of the canonical embedding be from
the best one? Such a characterization, even if only applicable to special primitives, would allow to lower
bound their leakage and would also help to understand the power of
two-party quantum cryptography in a more concise way.

A very natural generalization of our approach
would be  to see what happens when
 \emph{\correctness}\ is relaxed. 
\begin{conjecture}
Any correct protocol for $P_{X,Y}$ leaks as much as a \correct\ 
protocol for $P_{X,Y}$.
\end{conjecture}
The most obvious relaxation would be to consider as correct any $\ket{\psi}\in {\cal H}_{AB}\otimes 
{\cal H}_{A'B'}$ that produces $(x,y)$
with probability $P_{X,Y}(x,y)$ when registers $A$ and $B$ are measured but
 registers $A'$ and $B'$ can provide extra information 
about $Y$ and $X$ respectively. Remember that this is equivalent
to allowing for the quantum Markov chain 
conditions $A'\leftrightarrow X \leftrightarrow Y$ and 
$B'\leftrightarrow Y \leftrightarrow X$ not to
hold anymore. Would it be possible to find such a $\ket{\psi}$
with the property that for any regular embedding $\ket{\phi}\in \emb{P_{X,Y}}$:
\[ \Delta_{\psi}(P_{X,Y})  < \Delta_{\phi}(P_{X,Y})\,\,?
\]
A positive answer would reveal that some primitive $P_{X,Y}$ 
may be implemented with minimum leakage when viewed as a marginal in some
\emph{larger} probability distribution $P_{XX',YY'}$. A negative answer
would rather show that all our results hold unaffected for the standard notion
of correctness. Note however that the leakage is no more symmetric for the  standard 
notion of correctness.

It would also be interesting to find a measure of cryptographic
non-triviality for two-party primitives and  see how it relates to
the minimum leakage of any implementation by quantum protocols. For
instance, is it true that quantum protocols for primitive $P_{X,Y}$
leak more if the \emph{distance} between $P_{X,Y}$
and any trivial primitive increases?  

Another question we leave for future research is to define and
investigate other notions of leakage, e.g.~in the one-shot setting
instead of in the asymptotic regime (as outlined in
Footnote~\ref{foot:guess}). Results in the one-shot setting have
already been established for data compression~\cite{RW05}, channel
capacities~\cite{RWW06}, state-merging~\cite{WR07,Berta08} and other
(quantum-) information-theoretic tasks.

Furthermore, it would be interesting to find more applications for the
concept of leakage, considered also for protocols using an environment
as a trusted third party. In this direction, we have shown in
Theorem~\ref{no-amplif} that any two-party quantum protocol for a
given primitive, using a black box for an ``easier'' primitive, leaks
information. Lower-bounding this leakage is an interesting open
question. We might also ask how many copies of the ``easier''
primitive are needed to implement the ``harder'' primitive by a
quantum protocol, which would give us an alternative measure of
non-triviality for two-party primitives.

The approach used in this paper cannot easily be applied to 
cryptographic primitives modeled by unitary transforms.
Our approach is specialized to deal with classical primitives. 
It is an open question to determine the leakage of protocols implementing  
some unitary primitive. The few impossibility proofs for unitary
primitives  that we are aware of 
simply establish that perfect privacy cannot be achieved.
For example, it is shown in~\cite{DNS10}\ that quantum {\sf SWAP} 
is impossible (in fact, any unitary that never allows any of the party
to recover their input state). 
It would be very interesting to investigate the landscape of possibilities and impossibilities for unitary  primitives and see how it relates to the one for classical primitives. These two worlds might be very different\footnote{See~\cite{FKSZZ13} for a recent classification result for quantum protocols of classical cryptographic primitives.}.

\section*{Acknowledgements}
We would like to thank an anonymous referee for pointing out several shortcomings in earlier versions of this paper. LS is supported by the Danish Natural Science Research Council project QUSEP, and Canada's NSERC discovery grant. 
CS is supported by a NWO VENI project.

\bibliographystyle{alpha}

\bibliography{crypto,qip,procs}

\appendix

\section{Leakage of Universal Primitives}
\label{app:universalleakage}

\subsection{Exact calculations}
First, we look at the leakage of the embeddings of Rabin String OT (\srot). 

\begin{theorem} \label{thm:psrot} Any embedding of \psrot\ is at least
 $(1-O(r2^{-r}))$-leaking. For $r=1$ any embedding is at least
 $(h(\frac{1}{4})-\frac{1}{2}) \approx 0.311$-leaking. Furthermore, 
the leakage is the same  for all embeddings of \psrot.
\end{theorem}
\begin{proof}
Let
$$\ket{\psi}=\frac{1}{2^{\frac{r+1}{2}}}\sum_{x\in\{0,1\}^r}e^{i\theta(x,x)}\ket{xx}+\frac{1}{2^{\frac{r+1}{2}}}\left(\sum_{x\in\{0,1\}^r}e^{i\theta(x,\bot)}\ket{x}\right)\ket{\bot}\, ,$$
where $\bot$ denotes an erasure, be a general form of an embedding of \psrot. 

Define $\ket{\varphi} \assign \frac{1}{2^{r/2}}\sum_{x\in\{0,1\}^r}e^{i\theta(x,\bot)}\ket{x}$. If Bob 
receives the value of Alice's string successfully, Alice gets an ensemble $\rho^0=\frac{1}{2^r}\sum_{x\in\{0,1\}^r}\ketbra{x}{x}$. If an erasure occurs on Bob's side, Alice gets $\rho^1=\ketbra{\varphi}{\varphi}$. We find $S(A)$ by computing the eigenvalues of $\rho_A \assign \frac{1}{2}(\rho^0+\rho^1)$.

Since $\rho^0=\frac{1}{2^r}\id_A$, $\ket{v}$ is an eigenvector of $\rho_A$ if and only if it is an eigenvector of $\rho^1$. If $\ket{v}$ is an eigenvector of $\rho^1$ then either a) $\ket{v}=e^{i\theta}\ket{\varphi}$ or b) $\bracket{v}{\varphi}=0$. If a) is the case, then
$$\rho_A\ket{v}=\frac{1}{2}(\rho^0\ket{v}+\rho^1\ket{v})=\frac{1}{2}\left(1+\frac{1}{2^r}\right)\ket{v}\, ,$$
whereas in the case b),
$$\rho_A\ket{v}=\frac{1}{2}(\rho^0\ket{v}+\rho^1\ket{v})=\frac{1}{2^{r+1}}\, .$$
The state $\rho_A$ has eigenvalues $\{\frac{1}{2}+\frac{1}{2^{r+1}},\frac{1}{2^{r+1}}\}$, where $\frac{1}{2^{r+1}}$ has multiplicity $2^r-1$. $S(A)$ can then be computed as follows:
\begin{eqnarray*}
S(A)&=&-\left(\frac{1}{2}+\frac{1}{2^{r+1}}\right)\log\left(\frac{1}{2}+\frac{1}{2^{r+1}}\right)+\frac{2^r-1}{2^{r+1}}(r+1)\\
&= &\left(\frac{1}{2}+\frac{1}{2^{r+1}}\right)\left(1-\frac{1}{\ln{2}\cdot 2^{r}}+o\left(\frac{1}{2^r}\right)\right)+\frac{r+1}{2}-\frac{r+1}{2^{r+1}}= \frac{r}{2}+1-O\left(\frac{r}{2^{r}}\right).
\end{eqnarray*}

Since $I(X;Y)=\frac{r}{2}$, for the leakage we get:
$$\Delta_{\psi}(\psrot)=S(A)-I(X;Y)= 1- O\left(\frac{r}{2^{r}}\right)\, .$$
As we can see, the leakage does not depend on the phase-function $\theta$.
\qed

\end{proof}

In the following theorem we minimize the leakage of an embedding of $\pot$.
\begin{theorem}
\label{thm:pot}
Any $\ket{\psi}\in \emb{\pot}$ is at least $\frac{1}{2}$-leaking. 
The leakage is minimized by the canonical embedding.
\end{theorem}
\begin{proof}
Let
$$\ket{\psi}=\frac{1}{2\sqrt{2}}\sum_{x_0,x_1,c\in\{0,1\}}e^{i\theta(x_0x_1,cx_c)}\ket{x_0x_1}\ket{cx_c}$$
be a regular embedding of \pot. Without loss of generality assume that $\theta(00,00)=0$.  Notice that for the local phase-changing transformations
\begin{eqnarray*}
U^A&\assign&\ketbra{00}{00}+{\rm exp}(i\theta(01,00))\ketbra{01}{01}+{\rm   exp}(i(\theta(10,10)-\theta(00,10)))\ketbra{10}{10}\\
&+&{\rm exp}(i(\theta(10,10)+\theta(11,01)-\theta(00,10)-\theta(10,01)))\ketbra{11}{11},\\
U^B&\assign&\ketbra{00}{00}+{\rm exp}(i(\theta(00,10)+\theta(10,01)-\theta(10,10)))\ketbra{01}{01}\\
&+&{\rm exp}(i\theta(00,10))\ketbra{10}{10}+{\rm exp}(i(\theta(01,11)-\theta(01,00)))\ketbra{11}{11},
\end{eqnarray*}
we get
$$U^A\otimes U^B\ket{\psi}=\ket{\psi'}=\frac{1}{2}(\ket{0+}\ket{00}+\ket{1+}\ket{01}+\ket{+0}\ket{10}+\frac{\ket{0}+e^{i\omega}\ket{1}}{\sqrt{2}}\ket{1}\ket{11})\, ,$$
where $\omega=\theta(00,10)+\theta(01,00)+\theta(10,01)+\theta(11,11)-\theta(01,01)-\theta(10,10)-\theta(11,01).$

Let $A'$ denote Alice's quantum system for Alice and Bob sharing
$\ket{\psi'}$. Since $S(A)=S(A')$, we can minimize $S(A')$ in
order to minimize $S(A)$.  Assume that Alice and Bob share $\ket{\psi'}$.  For
Bob's selection bit $c=0$, Alice gets an ensemble
$\rho_0=\frac{1}{2}(\ketbra{0+}{0+}+\ketbra{1+}{1+})$, whereas for $c=1$, she
gets
$\rho_1=\frac{1}{2}(\ketbra{+0}{+0}+(\ket{01}+e^{i\omega}\ket{11})(\bra{01}+e^{-i\omega}\bra{11}))$,
where
$\rho_{A'}=\frac{1}{2}(\rho_0+\rho_1)$. By solving the characteristic
equation of $\rho_{A'}$ we get the set of eigenvalues
$\{\frac{1}{4}(1\pm\cos\frac{\omega}{4}),\frac{1}{4}(1\pm
\sin\frac{\omega}{4})\}$.  $S(A')$ can then be expressed as follows:
$$S(A')=1+\frac{h(\frac{1-\cos(\omega/4)}{2})+h(\frac{1-\sin(\omega/4)}{2})}{2}\, .$$
By computing the second derivative of $f(x)=h(\frac{1-\sqrt{x}}{2})$, we get that $f''(x)\leq 0$ in $[0,1]$, implying that $f$ is concave in $[0,1]$.
For $\alpha\in[0,1]$, Jensen's inequality yields $\frac{f(0)+f(1)}{2}\leq f(\alpha)$, and therefore,
$\frac{f(0)+f(1)}{2}\leq \frac{f(\alpha)+f(1-\alpha)}{2}$.
Consequently, the minimum of $h(\frac{1-\cos(\omega/4)}{2})+h(\frac{1-\sin(\omega/4)}{2})=f(\cos^2\frac{\omega}{4})+f(\sin^2\frac{\omega}{4})$ is achieved for $\omega=0$ and in this
case, $S(A')=\frac{3}{2}$.

Finally, we can conclude that the leakage is minimal for the canonical embedding and
$\Delta_{\psi}(P_{X,Y})=S(A)-I(X;Y)=S(A')-I(X;Y)\geq \frac{3}{2}-1=\frac{1}{2}$.
\qed

\end{proof}

There is also a more direct way to
interpret this quantity in the case of the canonical embedding
$\ket{\psi_{\vec{0}}}$ for \pot: If Alice and Bob share a single copy of
$\ket{\psi_{\vec{0}}}$ then there exist POVMs for both of them which
reveal Bob's selection bit to Alice, and the XOR of Alice's bits to
Bob, both with probability $\frac{1}{2}$.
Let $\ket{\Phi^{\pm}}=\frac{1}{\sqrt{2}}(\ket{00}\pm\ket{11})$,
$\ket{\Psi^{\pm}}=\frac{1}{\sqrt{2}}(\ket{01}\pm\ket{10})$ denote the
Bell states, and $\ket{\pm} \assign \frac{1}{\sqrt{2}}(\ket{0}\pm\ket{1})$.
Observe that the canonical embedding $\ket{\psi_{\vec{0}}}$ of \pot\ can be expressed as follows:
\[ 
\ket{\psi_{\vec{0}}} = \frac{1}{2}\ket{\Psi^{-}}\otimes \frac{\ket{\Psi^{-}}-\ket{\Phi^{-}}}{\sqrt{2}} +
           \frac{1}{2}\ket{\Phi^{-}}\otimes \frac{\ket{\Psi^{+}}-\ket{\Phi^{+}}}{\sqrt{2}} +
            \frac{1}{\sqrt{2}} \ket{++}\ket{++}.
\]
In order to get the value $x_0\oplus x_1$ of Alice's bits $x_0$ and $x_1$, 
Bob can use  POVM ${\sf B}=\{{\sf B}_0,{\sf B}_1,{\sf B}_?\}$ where
${\sf B}_0 \assign  
\frac{1}{2}(\ket{\Psi^-}-\ket{\Phi^-})(\bra{\Psi^-}-\bra{\Phi^-})$,
${\sf B}_1  \assign  
\frac{1}{2}(\ket{\Psi^+}-\ket{\Phi^+})(\bra{\Psi^+}-\bra{\Phi^+})$,
and
${\sf B}_{?}  \assign  \proj{++}$. It is easy to verify that Bob gets
outcome ${\sf B}_z$ for $z\in\{0,1\}$ (in which
case $x_0\oplus x_1= z$ with certainty) with probability $\frac{1}{2}$.
Alice's POVM can be defined as ${\sf A}=\{{\sf A}_0,{\sf A}_1,{\sf A}_?\}$ where
${\sf A}_0  \assign  \proj{-+}$,
${\sf A}_1  \assign  \proj{+-}$, and
${\sf A}_{?}  \assign  \id_2-{\sf A}_0-{\sf A}_1$.
By inspection we easily find that 
the probability for Alice to get Bob's selection bit 
is $1-\trace{({\sf A}_?\otimes\id_2)\proj{\psi_{\vec{0}}}} = \frac{1}{2}$.
For any regular embedding of \pot\ we can construct similar POVMs revealing the XOR 
of Alice's bits to Bob and Bob's selection bit to Alice with probability strictly 
more than $\frac{1}{4}$.

\subsection{Lower Bounds}
\begin{theorem}
\label{stringotleakage}
Any embedding $\ket{\psi}$ of \psot\ is  $(1-O(r2^{-r}))$-leaking.
\end{theorem}
\begin{proof}
We use Theorem~\ref{reduction} to show that any (regular) embedding of \psot\ leaks at least as much as some regular embedding of \psrot. Let $(A_0,A_1)$ and $B$ denote Alice's and Bob's respective registers.
Then $\ket{\psi}_{A_0A_1B}\in \emb{\psot}$ can be written in the form:
$$\ket{\psi}=\frac{1}{2^{r/2}}\sum_{x\in\{0,1\}^r} \ket{x}^{A_1}\ket{\psi^x}_{A_0B}\, ,$$ 
where each 
$$\ket{\psi^x}=\frac{1}{2^{(r+1)/2}}\sum_{x'\in\{0,1\}^r}\left(e^{i\theta(x',x,0)}\ket{x'}^{A_0}\ket{0,x'}^B+
e^{i\theta(x',x,1)}\ket{x'}_{A_0}\ket{1,x}^B\right)$$ 
can be viewed as a regular embedding of \psrot. According to Theorem~\ref{reduction} and Theorem~\ref{thm:psrot}, 
we get that 
$$\Delta_{\psot}\geq \Delta_{\psrot}=1-O(r/2^r)\, .$$
\qed
\end{proof}

\begin{theorem}
\label{noisyotleakage2}
If $p<\frac{1}{2}-\frac{1}{2\sqrt{2}} \approx 0.1464$ then 
$\Delta_{\potn}\geq \frac{\left(1/2-p-\sqrt{p(1-p)}\right)^2}{8\ln 2}.$
\end{theorem}
\begin{proof}
Before starting with the actual proof, we formulate a useful statement, relating two measures of uncertainty of 
a quantum ensemble. 

\begin{theorem}[Average Encoding Theorem~\cite{KNTZ01}]
\label{aver_encoding}
Let $E$ denote a quantum system storing the quantum part of a cq-state $\rho_{XE}= \sum_{x\in\mathcal{X}}P_X(x)\proj{x}\otimes \rho^x_E$. Then 
$$\sum_x P_X(x)\|\rho_E-\rho^x_E\|_1\leq\sqrt{2(\ln 2)S(X;E)}\, .$$
\end{theorem}

In order to prove Theorem~\ref{noisyotleakage2}, we first notice that for any regular embedding of $P_{X,Y_0Y_1}$ such that $Y_0$ and $Y_1$ are independent, it holds that
\begin{align} \label{jointentropy}
S(A;Y_0Y_1)\geq S(A;Y_0)+S(A;Y_1)\, .
\end{align}

We can write 
\begin{align*}
S(A;Y_0)+S(A;Y_1)&=H(Y_0)+H(Y_1)-S(Y_0|A)-S(Y_1|A)\\
&=H(Y_0Y_1)-S(Y_0|A)-S(Y_1|A)\\
&\leq H(Y_0Y_1)-S(Y_0Y_1|A)=S(A;Y_0Y_1) \, ,
\end{align*} 
which proves Inequality~\eqref{jointentropy}.

Let $X,Y_0,Y_1$ be random variables corresponding to Alice's pair of bits, Bob's selection bit, and its value, respectively. For $\potn$ we have that $I(X;Y_0Y_1)=1-h(p)$. As the selection bit $Y_0$ and the value $Y_1$ are independent, we can use~\eqref{jointentropy} to lower bound $S(A;Y_0Y_1)$ as follows
$$S(A;Y_0Y_1)\geq S(A;Y_0)+S(A;Y_1)\geq S(A;Y_0)+(1-h(p))\, .$$ 
Hence, for computing the lower bound on $S(A;Y_0Y_1)$, we only need 
to compute the lower bound on $S(A;Y_0)$. 
A state $\ket{\psi}\in \emb{\potn}$ can be written as 
$$\ket{\psi}=\frac{1}{\sqrt{2}}(\ket{\psi_0}^{AB_1}\ket{0}^{B_0}+\ket{\psi_1}^{AB_1}\ket{1}^{B_0})\, .$$ 

Let $\rho^0_A\assign \tr_{B_1}\proj{\psi_0}$ and $\rho^1_A\assign 
\tr_{B_1}\proj{\psi_1}.$  By applying Theorem~\ref{aver_encoding} from above, 
we get that 
$$\|\rho^0_A-\rho^1_A\|_1\leq \sqrt{8(\ln 2)S(A;Y_0)}\, ,$$
and therefore,
\begin{equation}
\label{averageencoding}
\frac{\|\rho^0_A-\rho^1_A\|_1^2}{8\ln 2}\leq S(A;Y_0).
\end{equation}
The trace norm of $\rho^0_A-\rho^1_A$ yields an upper bound on the entries of the matrix:
\begin{equation}
\label{norma}
|(\rho^0_A-\rho^1_A)_{ij}|\leq \|\rho^0_A-\rho^1_A\|_1.
\end{equation} 

We can write the state $\ket{\psi}$ in the form:
$$\ket{\psi}=\frac{1}{2}\sum_{y_0,y_1}\ket{\varphi^{y_0,y_1}}_A\ket{y_0,y_1}_{B_0B_1}\, ,$$ 
where 
\begin{eqnarray*}
\ket{\varphi_{0,y}}&=&\sqrt{\frac{1-p}{2}}\sum_{x=0}^1 e^{i\theta(y,x,0,y)}
\ket{y,x}_A\ket{0,y}_{B_0B_1}+\sqrt{\frac{p}{2}}\sum_{x=0}^1 e^{i\theta(y,x,0,1-y)}\ket{y,x}_A\ket{0,1-y}_{B_0B_1}\\
\ket{\varphi_{1,y}}&=&\sqrt{\frac{1-p}{2}}\sum_{x=0}^1 e^{i\theta(x,y,1,y)}
\ket{x,y}_A\ket{1,y}_{B_0B_1}+\sqrt{\frac{p}{2}}\sum_{x=0}^1 e^{i\theta(x,y,1,1-y)}\ket{x,y}_A\ket{1,1-y}_{B_0B_1}.
\end{eqnarray*} 
By evaluating the individual matrix entries of $(\rho^0_A-\rho^1_A)$ we get a simple lower bound on $|(\rho^0_A-\rho^1_A)_{ij}|$ for $i\neq j\in\{0,\dots,3\}$:\begin{equation}
\label{upboundentry}
|(\rho^0_A-\rho^1_A)_{ij}|\geq \frac{1-2p}{4}-\frac{\sqrt{(1-p)p}}{2}
\end{equation}
hence, from~(\ref{norma}) follows that 
$$\|\rho^0_A-\rho^1_A\|_1\geq  \frac{1-2p}{4}-\frac{\sqrt{(1-p)p}}{2}\, ,$$
yielding due to~(\ref{jointentropy}) and (\ref{averageencoding})  that 

$$S(A;Y_0Y_1)\geq 1-h(p)+S(A;Y_0)\geq  1-h(p)+\frac{(1/2-p-\sqrt{(1-p)p})^2}{32\ln 2}\, .$$ 

The lower-bound is non-trivial if $1/2-p-\sqrt{(1-p)p}>0$, 
which is true for $p<\frac{1}{2}-\frac{1}{2\sqrt{2}}$. 
The results yields the following lower-bound on the leakage of \potn: 
$$\Delta_{\potn} \geq \frac{(1/2-p-\sqrt{(1-p)p})^2}{32\ln 2}\, .$$ 
However, this lower-bound is very loose, since for $p=0$ we get 
that 
$$\Delta_{\pot}\geq \frac{1}{128\ln 2}\approx 0.011\, ,$$
which is much weaker than the optimal 
$$\Delta_{\pot}\geq \frac{1}{2}\, .$$ \qed
\end{proof}

It remains to mention that by using more careful analysis of the phases of $\ket{\varphi_{0.y}}$ and $\ket{\varphi_{1,y}}$, the lower bound on the absolute value of the outside-diagonal 
entries from~(\ref{upboundentry}) can be improved, yielding a non-trivial lower bound on the leakage for $p>0.1464$
and eventually, even for any $p<1/4$. 

\end{document}